\documentclass{article}
\PassOptionsToPackage{table}{xcolor}
\usepackage{arxiv}
\usepackage[T1]{fontenc}
\usepackage[round,authoryear]{natbib}
\usepackage{amsmath,amsfonts,bm}

\def\eqref#1{equation~\ref{#1}}
\def\1{\bm{1}}

\DeclareMathAlphabet{\mathsfit}{\encodingdefault}{\sfdefault}{m}{sl}
\SetMathAlphabet{\mathsfit}{bold}{\encodingdefault}{\sfdefault}{bx}{n}

\usepackage{hyperref}
\usepackage{url}
\usepackage{amsthm}
\usepackage{amssymb}

\usepackage[table]{xcolor}
\usepackage{algorithm}
\usepackage{algpseudocode}
\usepackage{soul}
\usepackage{svg}
\usepackage{graphicx}
\usepackage{mwe}  % provides example-image / example-image-a placeholders
\usepackage{filecontents}
\usepackage{diagbox}

\usepackage{tikz}
\usepackage{amsmath}

\usepackage{filecontents}

\usepackage{booktabs}    % For professional-looking tables
\usepackage{multirow}    % For merging rows
\usepackage{caption}
\usepackage{subcaption}
\usepackage{tabularx}

\usepackage{xurl} % This fixed the broken links
\usepackage{amssymb}

\usepackage{colortbl}
\usepackage{array}
\usepackage{threeparttable}
\usepackage{gensymb}
\definecolor{Gray}{gray}{0.9}
\definecolor{LightGreen}{RGB}{230,255,230}
\usepackage{xspace}
\usepackage{siunitx}
\usepackage{placeins}

\usepackage{tcolorbox}
\tcbuselibrary{breakable,skins}  % Added 'skins' library

\usepackage{hyperref}
\usepackage{makecell}

\usepackage{enumitem}
\setlist[itemize]{nosep, itemsep=1pt, topsep=1pt, leftmargin=1em, labelwidth=*}
\setlist[enumerate]{nosep, itemsep=1pt, topsep=1pt, leftmargin=1em, labelwidth=*}

\newcommand{\myparagraph}[1]{\smallskip\noindent\textbf{#1}\xspace}

\newcommand{\defensename}{\texttt{SABER}\xspace}

\usepackage[textsize=tiny,colorinlistoftodos,textwidth=15mm]{todonotes}

\newtheorem{theorem}{Theorem}
\newtheorem{proposition}[theorem]{Proposition}

\title{Don't Let One Lie Survive A Hundred Truths: A Selective Bayesian Trust Estimator for Collaborative Perception}
\author{
  Yutong Liu \quad Chenyi Wang \quad Ming F. Li \quad Qingzhao Zhang \\
  School of Electrical, Computing \& Software Engineering \\
  The University of Arizona, Tucson, AZ \\
  \texttt{\{yutongl, chenyw, lim, qzzhang\}@arizona.edu}
}
\date{}
\renewcommand{\headeright}{}
\renewcommand{\undertitle}{}
\renewcommand{\shorttitle}{SABER: Selective Bayesian Trust Estimation}

\begin{document}
\maketitle

\begin{abstract}
Collaborative perception (CP) enables connected vehicles to see beyond their own sensors but makes them dependent on messages they cannot independently verify. A compromised collaborator can surgically conceal a single safety-critical object or inject a non-existing one while correctly reporting many others. Existing Bayesian trust mechanisms pool agreement across objects, which, while effective against blatant untargeted attacks, either incurs high false-positive rates (FPR), or allows unrelated correct reports to dilute persistent attack evidence for stealthy single-object attackers.
To address this problem, we propose \defensename, a selective \textit{two-tier} Bayesian trust estimator. The first tier maintains broad agent and object trust, preserving the ability to downweight benign but low-quality contributors. Cumulative-sum screening selects agent--object pairs with persistent omissions or unsupported reports for focused Bayesian assessment. The second tier checks these pairs against other agents' evidence and maintains a separate, reference-weighted Beta state for each. The lowest pair score constrains agent trust, preventing unrelated reports from diluting a targeted attack.
We establish sufficient conditions for stronger attacker-side trust reductions with bounded additional benign false alarms at fixed thresholds. Compared with state-of-the-art CP defenses, \defensename improves attack detection while reducing benign FPRs. On OPV2V, \defensename improves defense ROC-AUC over MATE by up to 0.427 in late fusion and 0.337 in intermediate fusion. Against advanced intermediate-fusion data fabrication attacks, it increases detection rates over ROBOSAC and LUCIA by up to 96.40 and 67.07 percentage points, respectively, while reducing FPRs.
\end{abstract}

\section{Introduction}
\label{sec:introduction}

Collaborative perception (CP) helps connected and autonomous vehicles (CAVs) overcome sensing-range and occlusion limits by sharing measurements, features, or detected objects \citep{xu2022opv2v}. F-Cooper and V2VNet fuse complementary features \citep{chen2019f,wang2020v2vnet}, while subsequent methods refine information selection and fusion through learned collaboration graphs, attention, and confidence-guided communication \citep{li2021learning,xu2022v2x,hu2022where2comm}.

Shared information also creates a security risk: compromised collaborators can manipulate intermediate features or detected objects to hide real objects or introduce nonexistent ones \citep{zhang2024data}. Even one corrupted object can threaten safety: omitting one approaching vehicle can compromise an unprotected left turn even when other reports are correct \citep{Coopernaut}, and targeted attacks have been shown to affect braking and collision outcomes \citep{zhang2024data}.

Existing defenses identify suspicious collaborators through inconsistencies in shared perception. ROBOSAC checks collaborator subsets against an individual reference \citep{li2023among}, while MADE compares individual and collaborative predictions \citep{zhao2023malicious}. LUCIA and CP-Guard+ assess feature inconsistencies \citep{wang2025threat,hu2025cp}, and GCP examines spatial discrepancies and temporal motion anomalies \citep{tao2026gcp}. Bayesian trust estimation accumulates such evidence over time. MATE recursively updates agent and track trust from reported and expected observations to downweight low-quality contributors \citep{hallyburton2025security}.

However, pooling evidence across objects into one agent trust state allows unrelated correct reports to offset evidence of a targeted attack. This \emph{cross-object evidence dilution} lets an agent retain high trust while corrupting one or a few safety-critical objects. Increasing pooled-update sensitivity also penalizes benign errors more heavily, risking excessive false positives.

Object-level temporal structure offers a way to distinguish these attacks from ordinary benign errors. Sustaining a false scene description across frames can prolong its influence on downstream driving decisions \citep{zhang2024data,hallyburton2025security}. We therefore study defenses against persistent attacks on one or a few objects. Our goal is to preserve evidence of these persistent contradictions in trust estimation while tolerating ordinary benign errors.

Based on the insights, we introduce \defensename, \emph{Selective Agent--object Bayesian Estimation of Reliability}. \defensename includes a two-tier trust estimation, where the first maintains broad Bayesian estimates of each agent's reporting reliability and each tracked object's existence, using observations across objects and time. The second tier focuses on \emph{agent--track pairs}, where each represents one agent's reporting history about an object hypothesis tracked across frames. A cumulative-sum (CUSUM) statistic identifies pairs with persistent omissions or unmatched reports \citep{page1954continuous}. These pairs receive further Bayesian assessment rather than being labeled malicious by the screening statistic alone. Other agents' observations provide an object-existence reference, with the evaluated agent excluded from both current evidence and historical reference-trust updates.

Each selected agent--track pair receives a separate Bayesian trust state, represented by a Beta distribution \citep{josang2002beta}. Evidence about one track therefore cannot directly cancel evidence about another. Contradictions receive greater weight when other agents provide strong evidence against the assessed report or omission. Recent contradictions observed during screening are retained so that trust can respond promptly when the pair enters the second tier, without waiting for new evidence. The lowest assessed pair score constrains the agent's overall trust, preserving evidence of a localized attack even when most reports remain correct.

Our contributions are:
\begin{itemize}
\item \textbf{Identifying cross-object evidence dilution.} We identify a structural limitation of pooled agent-level trust estimation: correct reports about unrelated objects can mask persistent errors on a safety-critical object. Our analysis shows how the influence of a localized attack depends on the volume of unrelated evidence, rather than only on the strength of the contradiction. Because both accumulate over time, longer observation alone need not resolve the problem.
\item \textbf{Object-isolated Bayesian trust estimation.} We develop \defensename, a two-tier estimator that separates overall agent reliability from reliability on individual objects. Separate evidence states for agent--track pairs, combined through lowest-pair aggregation, prevent correct reports about unrelated objects from directly canceling a localized contradiction.
\item \textbf{Evaluation across fusion settings and attack durations.} We extensively evaluate \defensename on OPV2V benchmark~\citep{xu2022opv2v} under state-of-the-art (SOTA) object removal and spoofing attacks~\citep{zhang2024data,wang2025threat} in both late and intermediate fusion. We compare against SOTA defenses MATE~\citep{hallyburton2025security}, ROBOSAC~\citep{li2023among}, and LUCIA~\citep{wang2025threat}. Relative to MATE, \defensename improves ROC-AUC by up to 0.427 in late fusion and 0.337 in intermediate fusion. At the evaluated fixed thresholds, \defensename's detection of intermediate-fusion Data Fabrication (DF) spoofing attacks~\citep{zhang2024data} exceeds ROBOSAC and LUCIA by up to 96.40 and 67.07 percentage points, respectively, while incurring lower false-positive rates.
\end{itemize}

\section{Related Work}
\label{sec:related_work}

\myparagraph{Collaborative perception (CP).} CP expands sensing coverage by combining observations across vehicles through early, intermediate, or late fusion of raw sensor data, learned features, or object reports, respectively. Representative feature-fusion methods include Attentive Fusion in OPV2V \citep{xu2022opv2v}, DiscoNet's distillation-based collaboration \citep{li2021learning}, V2X-ViT's cross-agent attention \citep{xu2022v2x}, and Where2comm's confidence-guided communication \citep{hu2022where2comm}. Object-level sharing is supported by sensor-sharing standards, including SAE J3224's Sensor Data Sharing Message (SDSM) and ETSI's Collective Perception Message (CPM) \citep{SAEJ3224_2022,ETSI_TS103324_2023}. \defensename assesses object-report reliability in late fusion and extends this assessment to intermediate fusion by independently decoding each sender's received features.

\myparagraph{Attacks on shared perception.} \cite{tu2021adversarial} first explored untargeted attacks against CP to degrade accuracy. Subsequently, Data Fabrication studies targeted attacks through feature manipulations for spoofing and removal \citep{zhang2024data} of single objects. SOMBRA exploits attentive fusion for stealthy targeted and mass removal attacks \citep{wang2025threat}. These attacks motivate assessing sparse object-level inconsistencies induced by shared messages.

\myparagraph{Defenses and trust estimation.} ROBOSAC checks randomly sampled collaborator groups against a trusted receiver's individual detections \citep{li2023among}. MADE tests ego--collaborator consistency with conformal p-values and multiple testing \citep{zhao2023malicious}, while CAD checks fabrication through multi-vehicle observations \citep{zhang2024data}. LUCIA continuously reweights attention using feature-disagreement scores \citep{wang2025threat}, CP-Guard+ separates benign and malicious features \citep{hu2025cp}, and GCP combines spatial consistency with temporal motion-flow analysis \citep{tao2026gcp}. MATE recursively estimates agent and track trust for security-aware fusion \citep{hallyburton2025security}. \defensename preserves a separate evidence state for each assessed pair, preventing unrelated correct reports from directly replenishing its trust.

\myparagraph{Bayesian evidence and sequential screening.} Beta reputation models represent feedback through evidence counts with weighting and forgetting \citep{josang2002beta}. CUSUM accumulates evidence of sustained changes \citep{page1954continuous}, while generalized Bayes provides a loss-based interpretation of weighted updates \citep{bissiri2016general}. \defensename combines these principles by screening agent--track histories, assessing selected pairs against other agents' evidence, and retaining separate weighted states. The contribution lies in this selective, object-specific assessment rather than Bayesian updating or temporal screening alone.

\section{Problem Formulation}
\label{sec:problem_formulation}

\myparagraph{System model.} Let $j$ index an object and $t$ a frame. For each timestep, the ego agent as the trusted aggregator (AGG), receives messages from collaborating agents $i\in\mathcal I$, and produces a fused perception representation for downstream tasks. The AGG transforms object reports into a common coordinate system and associates reports of the same object across agents and frames into a track. Each message contains the collaborator pose $p_i^t$, declared fields of view (FOVs) $\widehat V_{ij}^t$, and either (1) object reports for late fusion, or (2) bird's-eye-view (BEV) features for intermediate fusion. 

While each message has an identifiable sender, no collaborating agent is assumed to be a trusted reference. If agent $i$ reports a box matched to track $j$, AGG treats it as evidence that the object exists. If the received message contains no matching box, AGG treats it as evidence against the object's existence only when agent $i$ declares the object's location visible. An unavailable message or an unmatched track outside the declared FOV provides no such negative evidence.

\myparagraph{Threat model.} One compromised agent (\textit{i.e.}, the attacker) aims to induce unsafe driving of a victim collaborator by suppressing real objects (\textit{i.e.}, removal attack) or reporting nonexistent ones (\textit{i.e.}, spoofing attack) in CP over frames $t_s,\ldots,t_e$. It targets one or a few objects by modifying object reports in late fusion or perturbing feature messages in intermediate fusion, and may falsify its visibility declarations. It cannot modify other agents' messages or AGG. We exclude denial-of-collaboration, where an agent withholds observations by falsely declaring an object outside its FOV. Such behavior reduces the information available for collaboration rather than supplying a false observation within the declared FOV, and falls outside our threat model. The defense knows the fusion setting but not the attacker's identity or attack interval. A target can be assessed even if the attacker never reports it, provided other agents supply relevant observations.

\myparagraph{Benign faults vs. attacks.} Safety-oriented attacks can differ from ordinary benign perception errors in two key respects. \emph{Spatially}, attacks target nearby, decision-critical objects, whereas benign errors are more prevalent farther from the observing vehicle. Our single-agent measurements with the intermediate-fusion detector on OPV2V test scenarios show that most erroneous boxes lie beyond 50\,m (Figure~\ref{subfig:benign_distance}). \emph{Temporally}, benign errors are predominantly short-lived (Figure~\ref{subfig:benign_duration}), whereas attacks can sustain errors across frames to affect tracking and driving decisions. For example, \citet{jia2019fooling} analyze a tracker operating at 30 Hz, under which deleting an established track through timeout requires 60 consecutive missed detections. In a CARLA unprotected-left-turn scenario (Figure~\ref{fig:safety_scenario}), continuous removal of an oncoming vehicle from the victim's obstacle input reduces time to collision (TTC) to zero and results in a collision (Figure~\ref{subfig:continuous_attack_ttc}). This illustrates how persistent perception errors can cause serious safety consequences without corrupting many objects. A five-frame example with benign errors and designated attack positions is provided in Figure~\ref{fig:case_demo}. These spatial and temporal differences motivate our filters for selecting suspicious agent-object pairs.

The distance distribution counts each erroneous box in each frame once, whereas the duration distribution counts each uninterrupted error run once. Both pool false positives and false negatives across all 3--5-agent OPV2V test scenarios. The TTC curve follows successive instants of a single continuous removal rollout.

\begin{figure}[t]
    \centering
    \captionsetup[subfigure]{font=small,skip=3pt}
    \begin{subfigure}[b]{0.32\linewidth}
        \centering
        \includegraphics[width=\linewidth]{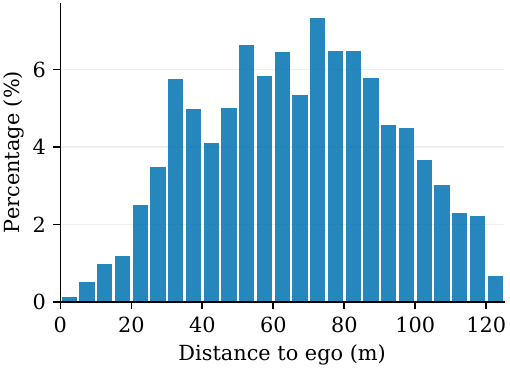}
        \caption{Benign error distance to ego agent.}
        \label{subfig:benign_distance}
    \end{subfigure}\hfill
    \begin{subfigure}[b]{0.32\linewidth}
        \centering
        \includegraphics[width=\linewidth]{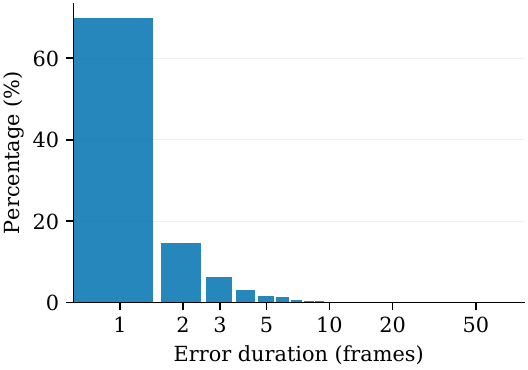}
        \caption{Benign error duration.}
        \label{subfig:benign_duration}
    \end{subfigure}\hfill
    \begin{subfigure}[b]{0.32\linewidth}
        \centering
        \includegraphics[width=\linewidth]{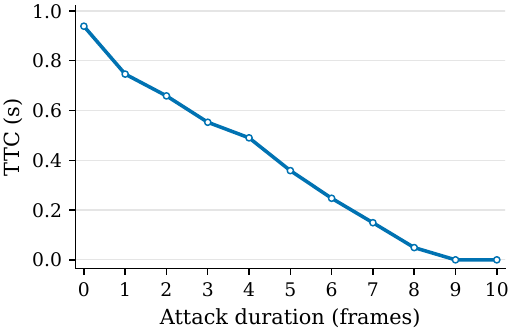}
        \caption{TTC under continuous removal attacks.}
        \label{subfig:continuous_attack_ttc}
    \end{subfigure}
    \caption{Statistics of benign errors and continuous attack.}
    \label{fig:benign_attack_motivation}
\end{figure}

\begin{figure}[t]
    \centering
    \includegraphics[width=0.40\linewidth]{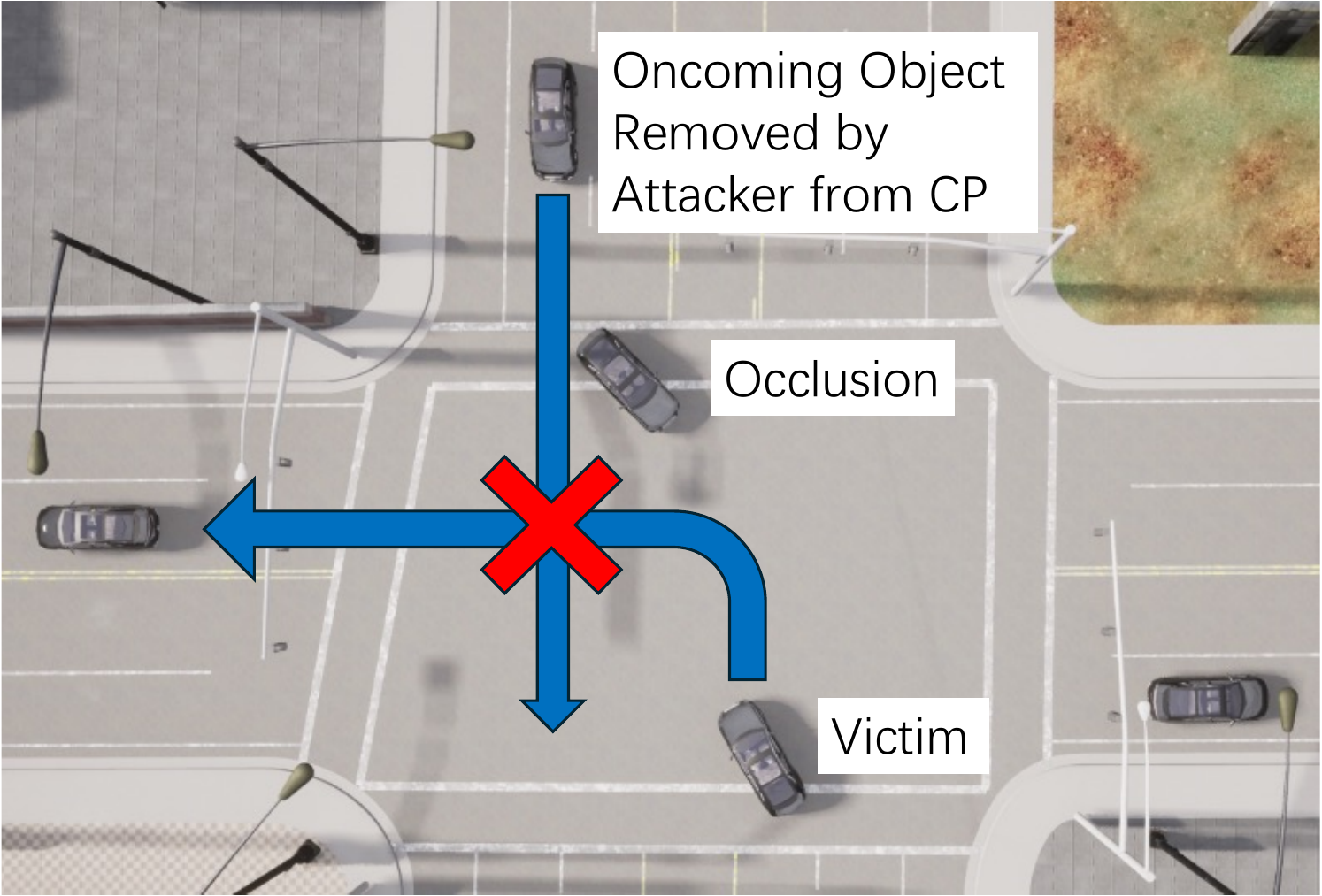}
    \caption{CARLA unprotected-left-turn scenario. Occlusion limits the victim's view of the oncoming vehicle targeted by the removal attack. Blue arrows indicate vehicle paths, and the red cross marks their conflict point.}
    \label{fig:safety_scenario}
\end{figure}

\begin{figure}[t]
    \centering
    \includegraphics[width=\linewidth]{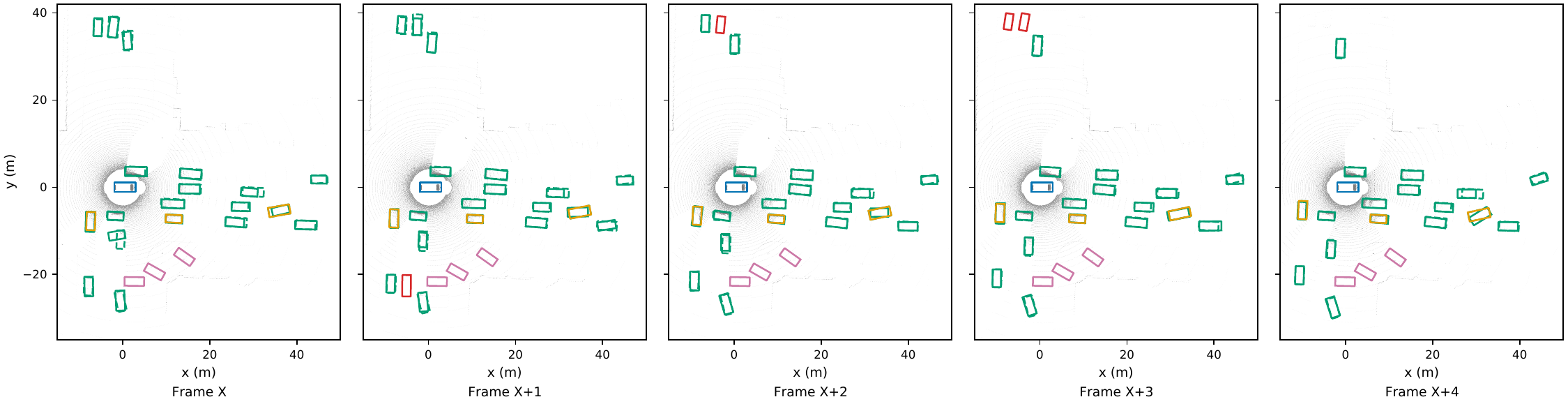}
    \caption{Five consecutive single-agent views. Blue boxes are ego agents. Green dashed boxes denote ground truth and green solid boxes denote matched detections. Red boxes mark unmatched FP/FN. Orange boxes mark designated removal targets. Purple boxes mark designated spoofing positions.}
    \label{fig:case_demo}
\end{figure}

\myparagraph{Objective of trust estimation.} Let $R_i^t\in[0,1]$ denote agent trust, with higher values indicating greater reliability. At the attack endpoint $t_\star=t_e$, define the true-positive rate (TPR) and false-positive rate (FPR) by
\begin{equation}
 \mathrm{TPR}=\Pr(R_a^{t_\star}<\eta),\qquad
 \mathrm{FPR}=\Pr_b(R_b^{t_\star}<\eta),
 \label{eq:detection_objective}
\end{equation}
where $a$ is the attacker, $b$ is a benign peer in the attacked scene, and $\eta$ is the fixed trust threshold. We aim to improve the detection--false-alarm trade-off over broad agent-level trust assessment.

\section{Methodology}
\label{sec:methodology}

\defensename uses two-tier Bayesian assessment to preserve localized attack evidence alongside overall agent reliability (Figure~\ref{fig:overview}). Tier 1 estimates broad agent trust $T_{1,i}$ across tracks using MATE's Bayesian estimator \citep{hallyburton2025security}. Tier 2 screens for persistent reporting anomalies involving nearby tracks and maintains a separate Bayesian trust state for each selected \emph{agent--track pair} $(i,j)$, representing agent $i$'s reporting history about track $j$. Each selected pair is evaluated using other agents' observations, with stronger peer support giving contradictory reports or omissions greater weight. When assessment begins, cached contradictions from the screening period are also used to update the pair's trust. The lowest assessed pair score constrains the final agent trust, preventing correct reports about other tracks from masking a localized contradiction. A pair leaves tier 2 when the agent's reporting about that track consistently agrees with other agents' observations.

\begin{figure}[t]
\centering
\includegraphics[width=\linewidth]{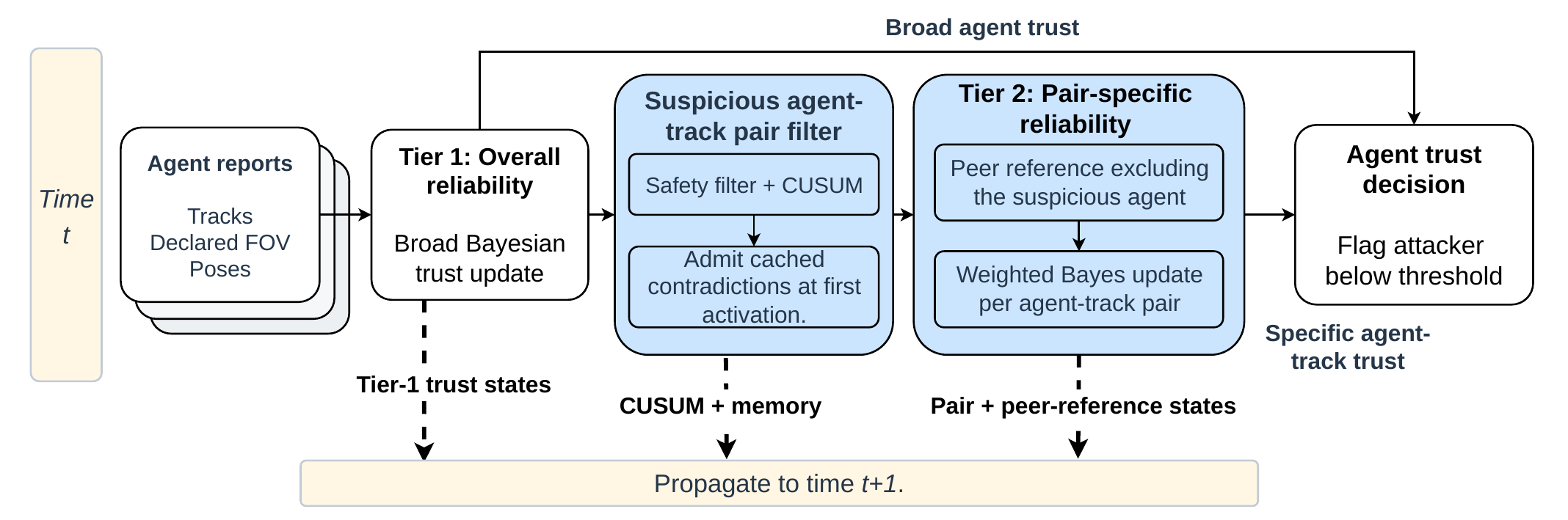}
\caption{Overview of \defensename.}
\label{fig:overview}
\end{figure}

\subsection{Spatial and Temporal Screening}
\label{subsec:screening}

\myparagraph{Spatial screening.} Objects far from every collaborating agent are harder to perceive accurately and are less likely to pose an immediate safety risk. We therefore restrict new tier-2 evidence to objects within $r_{\mathrm{safe}}$ of at least one agent to reduce false alarms caused by uncertain distant observations.\looseness=-1

\myparagraph{Temporal screening.} Track association links reports about the same object across frames, allowing CUSUM to accumulate repeated omissions or unmatched reports for each agent--track pair. Let $d_{ij}^t=1$ when agent $i$ reports a box associated with track $j$, and zero otherwise. For reference tracks, the false-negative (FN) screen treats a missing report as an event, $e_{ij}^t=1-d_{ij}^t$, only when the sender's message is available and the location is declared visible. Matched reports are non-events. For candidate tracks formed from unmatched boxes, the false-positive (FP) screen uses $e_{ij}^t=d_{ij}^t$ while the candidate remains unmatched to the reference track set.

A cumulative-sum (CUSUM) statistic detects events occurring more frequently than in benign histories \citep{page1954continuous,xie2021sequential}:
\begin{equation}
 \ell_{ij}^t=e_{ij}^t\log\frac{p_1}{p_0}+(1-e_{ij}^t)\log\frac{1-p_1}{1-p_0},\quad S_{ij}^t=\max\{0,S_{ij}^{t-1}+\ell_{ij}^t\}.
 \label{eq:cusum}
\end{equation}
Here $p_0$ is the smoothed event frequency estimated from eligible clean observations for each screen, and $p_1$ specifies the higher frequency to detect, with $0<p_0<p_1<1$. Starting from zero, $S_{ij}^t$ rises on events and falls toward zero on non-events, retaining evidence across occasional interruptions. A pair enters tier 2 when $S_{ij}^t\ge h$. To calibrate $h$, we take the maximum CUSUM over all relevant pairs and times in each clean agent--window and choose $h>0$ so that no more than a prespecified fraction of these windows triggers screening. FN and FP are calibrated separately. This limits empirical benign screening activations, not final agent false alarms, because crossing $h$ triggers assessment rather than an attacker flag. Appendix~\ref{subsec:screen_details} gives the exact calibration rule.

\subsection{Pair-Isolated Bayesian Trust Estimation}
\label{subsec:focused}

\myparagraph{Reference from other agents.} Let $q_{j,-i}^t\in[0,1]$ denote trust that track $j$ represents a real object, accumulated from other agents' eligible reporting histories since the track's introduction, with values near one supporting existence and values near zero supporting nonexistence. The reference starts from a neutral Beta prior and accumulates peer reports or eligible omissions, weighted by the reporting peers' reliability. Agent $i$ is excluded throughout the histories of both the object-existence reference and the peer-reliability estimates used to weight it. This prevents its reports from reinforcing their own reference, directly or through historical trust updates. Appendix~\ref{subsec:references} gives the reference update.

\myparagraph{Report consistency and evidence weight.} We compare the report with the reference through consistency $v_{ij}^t$ and assign evidence weight $w_{ij}^t$:
\begin{equation}
 v_{ij}^t=d_{ij}^tq_{j,-i}^t+(1-d_{ij}^t)(1-q_{j,-i}^t),\qquad
 w_{ij}^t=\lambda|2q_{j,-i}^t-1|,\quad\lambda\ge1.
 \label{eq:consistency_weight}
\end{equation}
High consistency means agreement with the reference. Both omitting an object supported by other agents and reporting an object contradicted by them give low consistency, allowing removal and spoofing to use the same update rule. The weight vanishes at a neutral reference and grows symmetrically toward strong support for existence or nonexistence. The multiplier $\lambda$ amplifies this selected evidence without dividing its influence among unrelated tracks. We further explain the weight's design properties in Appendix~\ref{subsec:references}.

\myparagraph{Memory of recent inconsistencies.} During screening, a bounded memory caches eligible contradictions within a fixed frame lookback, preserving the reference and weight computed when each observation arrived. When the pair first enters tier 2, previously unused contradictions from the same track-hypothesis episode contribute to its initial update, discounted by their age. The current frame is handled by the ordinary update rather than memory, so each observation contributes once. This enables a response without waiting for new contradictions, but applies no tier-2 penalty before activation. Appendix~\ref{subsec:memory} defines the lookback and admission rule.

\myparagraph{A separate Beta state for each pair.} For a selected pair, let $\alpha_{ij}^t$ and $\beta_{ij}^t$ represent positive and negative evidence, including the prior. Its trust is the mean of the corresponding Beta distribution \citep{josang2002beta}:
\begin{equation}
 B_{ij}^t=\frac{\alpha_{ij}^t}{\alpha_{ij}^t+\beta_{ij}^t},
 \label{eq:pair_mean}
\end{equation}
where the parameters are updated as
\begin{equation}
 \begin{aligned}
 \alpha_{ij}^t=\widetilde\alpha_{ij}^t+\Delta_{ij,+}^t+M_{ij,+}^t,\quad
 \beta_{ij}^t=\widetilde\beta_{ij}^t+\Delta_{ij,-}^t+M_{ij,-}^t.
 \end{aligned}
 \label{eq:tier2_update}
\end{equation}
The tildes denote temporally discounted states. Current evidence contributes $\Delta_{ij,\pm}^t=w_{ij}^t\phi_\pm(v_{ij}^t)$ when the pair is eligible and selected, and zero otherwise. The functions $\phi_+$ and $\phi_-$ map consistency to positive and negative evidence, respectively. The terms $M_{ij,\pm}^t$ admit cached evidence once at activation and are otherwise zero. Weighting the evidence increments in this way corresponds to a generalized Bayesian update \citep{bissiri2016general}.
At each pair's first evidence admission, initialize its Beta parameters as $(\kappa\mu_i,\kappa(1-\mu_i))$, where $\kappa>0$ is the total prior evidence mass and $0<\mu_i<1$. For agent $i$, $\mu_i$ is its tier-1 trust immediately before its first nonzero tier-2 evidence admission and remains fixed within the window. If admission occurs at the first frame, use the tier-1 prior mean. Later pairs share this mean but receive separate evidence states. Thus initialization transfers reliability without importing accumulated evidence about unrelated objects. A larger $\kappa$ makes the same new evidence less influential.

Before adding evidence, propagate each retained state toward its initial parameters:
\begin{equation}
 \widetilde\alpha_{ij}^t=(1-\rho_t)\alpha_{ij}^{t-1}+\rho_t\kappa\mu_i,\qquad
 \widetilde\beta_{ij}^t=(1-\rho_t)\beta_{ij}^{t-1}+\rho_t\kappa(1-\mu_i),
 \label{eq:propagation}
\end{equation}
where $\rho_t=\min(1,\Delta t/\tau_2)$, $\Delta t$ is elapsed time, and $\tau_2>0$ controls forgetting. This discounts past evidence while retaining the initial prior. An ineligible observation adds no evidence, but the retained state still propagates. Appendix~\ref{subsec:beta} defines the consistency increments and their generalized Bayesian interpretation.

\subsection{Agent Trust and Decision}
\label{subsec:decision}

Let $\mathcal A_i^t$ contain the track indices whose initialized pairs $(i,j)$ are currently under tier-2 assessment. When this set is nonempty, the lowest pair score constrains tier-2 agent trust:
\begin{equation}
 T_{2,i}^t=\min\left\{\mu_i,\min_{j\in\mathcal A_i^t}B_{ij}^t\right\}.
 \label{eq:pair_aggregate}
\end{equation}

We combine broad and focused assessment through
\begin{equation}
 R_i^t=\begin{cases}
 \min\{T_{1,i}^t,T_{2,i}^t\},&\mathcal A_i^t\ne\varnothing,\\
 T_{1,i}^t,&\text{otherwise}.
 \end{cases}
 \label{eq:combine}
\end{equation}
This preserves the constraint imposed by any assessed pair while retaining tier 1's broad reliability estimate. Sustained corroborated recovery removes a pair from $\mathcal A_i^t$. When no pair remains, the agent returns to tier-1 assessment without the fixed initial mean continuing to cap its score.

An agent is flagged when $R_i^t<\eta$, using one fixed threshold across agents, frames, and attack types. The trust threshold decides whether to flag an agent, whereas the CUSUM threshold $h$ decides which pairs receive additional assessment.

\myparagraph{Application to intermediate fusion.}
While \defensename is designed for the standard late fusion CP consistent with SAE SDSM or ETSI CPM messages \citep{SAEJ3224_2022,ETSI_TS103324_2023} using received object reports directly, it can be applied to the intermediate fusion paradigm \citep{xu2022opv2v,xu2022v2x}. First, AGG decodes each sender's received BEV feature message $F_i^t$ separately as $\mathcal B_i^t=D(F_i^t)$ for trust estimation with the shared detection decoder $D$. AGG then associates the decoded boxes $\mathcal B_i^t$ with tracks and applies the same two-tier trust estimator.

\subsection{Theoretical Properties}
\label{subsec:theoretical_properties}

We summarize the theoretical results below, with their assumptions and proofs in Appendix~\ref{sec:theory_statements}.

\myparagraph{Protection against cross-object dilution.} Pooled assessment allows unrelated correct reports to raise trust (Proposition~\ref{prop:dilution}). Separate pair states and minimum aggregation prevent this direct cross-object cancellation:
\begin{proposition}[Object-specific evidence isolation]
\label{prop:isolation_summary}
Fix a pair's initialization, reference and report histories, eligibility, screening, evidence admission, temporal discounting, and current assessment membership. Direct evidence updates for other tracks cannot change $B_{ij}$. While the pair remains assessed, $T_{2,i}\le B_{ij}$.
\end{proposition}

\myparagraph{Trust response to persistent contradictions.} The analysis characterizes memory's effect on pair trust (Proposition~\ref{prop:memory}) and gives sufficient weighted-evidence conditions for detection by attack completion (Proposition~\ref{prop:finite_detection}). Theorem~\ref{thm:paired} establishes when the expected attack-induced decrease exceeds broad assessment alone, accounting for any additional reduction on the paired clean input.

\myparagraph{Attacker--benign separation and false alarms.} Theorem~\ref{thm:separation} and Proposition~\ref{prop:threshold_interval} give conditions for attacker--benign trust separation and a fixed threshold combining high detection probability with bounded additional benign false alarms.

\section{Evaluation}
\label{sec:evaluation}

\subsection{Experimental Setup}
\label{subsec:evaluation_settings}

\myparagraph{Data and perception settings.} We evaluate on OPV2V scenes with three to five agents \citep{xu2022opv2v}. We consider late fusion, where agents share object reports~\cite{ETSI_TS103324_2023}, and intermediate fusion sharing BEV features. Both use PointPillars \citep{lang2019pointpillars} as the encoder backbone. Late-fusion reports come from independent single-agent detection. Intermediate experiments use the Attentive Fusion algorithm~\citep{xu2022opv2v}, with each received message decoded independently for trust assessment. MATE and \defensename use a BEV IoU threshold of 0.3 for report-to-track association, and ROBOSAC uses the same threshold for detection-set matching.

\myparagraph{Attack cases.} The full evaluation dataset contains 688 cases, 341 for removal cases and 347 for spoofing. A case specifies a window, an attacker, and three eligible target objects or spoofing positions, selected before defense evaluation. Removal targets must match the attacker's cached predictions throughout the main attack interval. Spoofing positions must overlap neither real objects nor each other and satisfy visibility requirements. We attack one, two, and three targets, respectively, to study sparse errors. Late-fusion attacks remove or insert boxes directly in the shared message. Intermediate-fusion attacks use projected gradient descent (PGD)~\citep{pgd18madry} and the white-box online adversarial attack of \citet{zhang2024data}, abbreviated DF here, for removal and spoofing. SOMBRA is evaluated for removal \citep{wang2025threat} only. Appendix~\ref{subsec:case_selection} provides more details about the attack target selection strategy.

\myparagraph{Experimental protocol.} We select 30-frame windows with a ten-frame stride and reset trust histories between windows. The main comparison uses zero-based frames 0--19, with ten benign frames followed by an attack between $t_s=10$ and $t_e=19$. We measure pre-attack behavior at frame 9 and attacker detection at attack completion at frame 19. The remaining frames support the duration study in Section~\ref{subsec:sensitivity}, with full results in Appendix~\ref{subsec:attack_duration_sensitivity}. In the main comparisons, fixed agent-trust thresholds are 0.1 for \defensename, 0.8 for MATE, and 0.1 for LUCIA. ROBOSAC uses a detection-set consistency threshold of 0.3. These values remain unchanged across attack types and target counts. The Bayesian trust evaluations are deterministic for the fixed inputs and configurations. Screening calibration follows Appendix~\ref{subsec:screen_details}.

\subsection{Baselines and Metrics}
\label{subsec:defense_comparisons}

\myparagraph{Baselines.} We compare with SOTA fusion defense MATE, which represents broad Bayesian trust estimation \citep{hallyburton2025security}. Intermediate-fusion experiments additionally compare ROBOSAC~\citep{li2023among}, which selects collaborators by thresholding agreement with a trusted receiver, and LUCIA~\citep{wang2025threat}, which scores collaborators through feature disagreement. Appendix~\ref{subsec:baseline_implementation} describes the baseline implementations and evaluation settings.

\begin{figure}[t]
\centering
\includegraphics[width=0.85\linewidth]{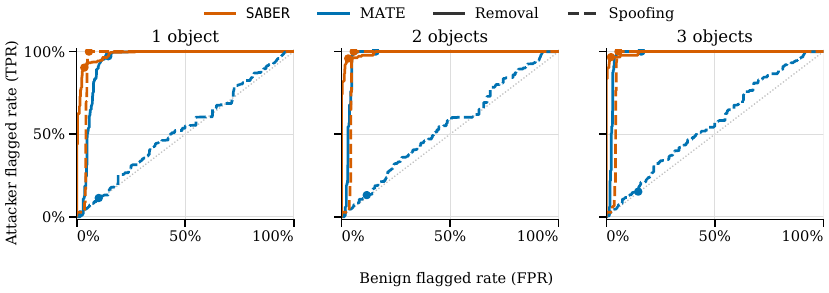}
\caption{Late-fusion ROC at attack completion. \defensename improves attacker–benign separation over MATE for both removal and spoofing, with especially large gains against spoofing.}
\label{fig:late_roc}
\end{figure}

\myparagraph{Metrics.} Following Eq.~(\ref{eq:detection_objective}), the true-positive rate (TPR), or \emph{attacker flagged rate}, is the rate of flagging the attacker. The false-positive rate (FPR), or \emph{benign flagged rate}, is the fraction of benign agents flagged. Higher TPR and lower FPR are better. We report FPR before attack onset and TPR/FPR at attack completion. We also report the area under the receiver operating characteristic curve (ROC-AUC) at attack completion to evaluate attacker--benign separation across thresholds. Higher AUC indicates better overall separation. Before attack onset, the designated attacker is benign and any flag on it is a false alarm. MATE, \defensename, and LUCIA are scored by thresholding their respective trust scores. ROBOSAC does not assign individual trust scores, so we count its rejection of a collaborator as a flag under the rule in Appendix~\ref{subsec:baseline_metrics}. Its TPR therefore measures attacker rejection rather than a separate identity classification.

\subsection{Main Results}
\label{subsec:main_results}

\myparagraph{Overall Detection Performance.} \defensename achieves a favorable detection--false-alarm balance across late and intermediate fusion while retaining low pre-attack false alarms (Table~\ref{tab:main_comparison}). MATE can attain higher removal detection, but at the cost of flagging substantially more benign agents, and remains weak against spoofing. \defensename improves low-FPR detection, especially for spoofing, and has higher ROC-AUC than MATE across the evaluated attack families and target counts (Figures~\ref{fig:late_roc} and~\ref{fig:intermediate_roc}, Appendix~\ref{sec:additional_results}). The improvement therefore extends beyond a single operating threshold. LUCIA is particularly effective against SOMBRA, but does not offer the same balance across the other attacks and pre-attack observations.

\begin{table}[t]
\centering
\begin{threeparttable}[t]
\scriptsize
\setlength{\tabcolsep}{2pt}
\begin{tabular}{@{}llrrrrrrrrrr@{}}
\toprule
\multirow{3}{*}{Attack} & \multirow{3}{*}{Method} & \multirow{2}{*}{Pre-attack} & \multicolumn{9}{c}{End-attack} \\ \cmidrule(lr){4-12} & & & \multicolumn{3}{c}{1 object} & \multicolumn{3}{c}{2 objects} & \multicolumn{3}{c}{3 objects} \\ \cmidrule(lr){4-6}\cmidrule(lr){7-9}\cmidrule(lr){10-12} & & FPR\% $\downarrow$ & TPR\% $\uparrow$ & FPR\% $\downarrow$ & AUC $\uparrow$ & TPR\% $\uparrow$ & FPR\% $\downarrow$ & AUC $\uparrow$ & TPR\% $\uparrow$ & FPR\% $\downarrow$ & AUC $\uparrow$ \\
\midrule
\multicolumn{12}{l}{\textit{Late fusion: bounding box perturbation}} \\
\rowcolor{Gray}\cellcolor{white} & MATE & 22.81 & 97.22 & 14.53 & 0.939 & 100.00 & 15.64 & 0.966 & 100.00 & 16.12 & 0.975 \\

\rowcolor{LightGreen}\cellcolor{white}\multirow{-2}{*}{Removal} & \defensename & \textbf{1.52} & 90.29 & 3.48 & \textbf{0.984} & 95.87 & 3.09 & \textbf{0.993} & 96.67 & 2.13 & \textbf{0.995} \\
\midrule
\rowcolor{Gray}\cellcolor{white} & MATE & 22.25 & 11.36 & 10.13 & 0.531 & 13.11 & 11.63 & 0.549 & 15.20 & 14.70 & 0.566 \\
\rowcolor{LightGreen}\cellcolor{white}\multirow{-2}{*}{Spoofing} & \defensename & \textbf{0.60} & 100.00 & 5.59 & \textbf{0.958} & 100.00 & 5.59 & \textbf{0.959} & 100.00 & 5.59 & \textbf{0.959} \\
\midrule
\multicolumn{12}{l}{\textit{Intermediate fusion: feature perturbation}} \\
\rowcolor{Gray}\cellcolor{white} & ROBOSAC & 2.93 & 3.81 & 3.47 & N/A & 3.74 & 3.42 & N/A & 3.67 & 3.42 & N/A \\
 & LUCIA & 27.49 & 31.09 & 27.47 & 0.529 & 33.80 & 26.54 & 0.556 & 35.58 & 25.49 & 0.580 \\
\rowcolor{Gray}\cellcolor{white} & MATE & 27.86 & 76.83 & 26.03 & 0.784 & 91.50 & 26.17 & 0.881 & 96.19 & 26.69 & 0.924 \\
\rowcolor{LightGreen}\cellcolor{white}\multirow{-4}{*}{\shortstack[c]{\strut PGD\\\strut Removal}} & \defensename & \textbf{1.27} & 48.68 & 1.32 & \textbf{0.819} & 67.16 & 1.17 & \textbf{0.909} & 77.42 & 1.17 & \textbf{0.945} \\
\midrule
\rowcolor{Gray}\cellcolor{white} & ROBOSAC & 2.93 & 3.89 & 3.52 & N/A & 3.81 & 3.49 & N/A & 3.96 & 3.49 & N/A \\
 & LUCIA & 27.49 & 30.01 & 27.66 & 0.522 & 32.55 & 26.69 & 0.544 & 34.26 & 25.95 & 0.563 \\
\rowcolor{Gray}\cellcolor{white} & MATE & 27.86 & 94.72 & 25.95 & 0.872 & 99.41 & 26.61 & 0.939 & 100.00 & 27.64 & 0.962 \\
\rowcolor{LightGreen}\cellcolor{white}\multirow{-4}{*}{\shortstack[c]{\strut DF\\\strut Removal}} & \defensename & \textbf{1.27} & 85.92 & 1.32 & \textbf{0.963} & 92.38 & 1.03 & \textbf{0.986} & 93.84 & 0.88 & \textbf{0.992} \\
\midrule
\rowcolor{Gray}\cellcolor{white} & ROBOSAC & 2.93 & 4.84 & 3.67 & N/A & 4.99 & 3.67 & N/A & 5.72 & 3.96 & N/A \\
 & LUCIA & 27.49 & 90.91 & 3.27 & \textbf{0.993} & 98.09 & 0.59 & \textbf{1.000} & 99.56 & 0.07 & \textbf{1.000} \\
\rowcolor{Gray}\cellcolor{white} & MATE & 27.86 & 97.07 & 26.74 & 0.897 & 99.12 & 26.88 & 0.942 & 99.71 & 27.98 & 0.959 \\
\rowcolor{LightGreen}\cellcolor{white}\multirow{-4}{*}{\shortstack[c]{\strut SOMBRA\\\strut Removal}} & \defensename & \textbf{1.27} & 85.92 & 1.61 & 0.965 & 92.38 & 1.17 & 0.985 & 93.84 & 0.88 & 0.991 \\
\midrule
\rowcolor{Gray}\cellcolor{white} & ROBOSAC & \textbf{0.00} & 1.59 & 1.63 & N/A & 1.59 & 1.61 & N/A & 1.59 & 1.61 & N/A \\
 & LUCIA & 24.87 & 24.78 & 28.24 & 0.468 & 24.93 & 27.95 & 0.473 & 25.36 & 27.95 & 0.477 \\
\rowcolor{Gray}\cellcolor{white} & MATE & 22.25 & 27.09 & 21.13 & 0.575 & 35.45 & 23.01 & 0.616 & 42.94 & 23.90 & 0.652 \\
\rowcolor{LightGreen}\cellcolor{white}\multirow{-4}{*}{\shortstack[c]{\strut PGD\\\strut Spoofing}} & \defensename & \textbf{0.00} & 55.62 & 1.44 & \textbf{0.819} & 76.66 & 1.44 & \textbf{0.910} & 85.01 & 1.44 & \textbf{0.947} \\
\midrule
\rowcolor{Gray}\cellcolor{white} & ROBOSAC & \textbf{0.00} & 1.87 & 1.87 & N/A & 1.87 & 1.87 & N/A & 1.87 & 1.87 & N/A \\
 & LUCIA & 24.87 & 27.09 & 27.38 & 0.491 & 29.32 & 26.22 & 0.513 & 31.20 & 25.29 & 0.533 \\
\rowcolor{Gray}\cellcolor{white} & MATE & 22.25 & 34.29 & 20.70 & 0.629 & 46.40 & 22.72 & 0.693 & 61.10 & 25.34 & 0.760 \\
\rowcolor{LightGreen}\cellcolor{white}\multirow{-4}{*}{\shortstack[c]{\strut DF\\\strut Spoofing}} & \defensename & \textbf{0.00} & 87.61 & 1.44 & \textbf{0.966} & 96.25 & 1.30 & \textbf{0.989} & 98.27 & 1.30 & \textbf{0.993} \\
\bottomrule
\end{tabular}
\begin{tablenotes}[flushleft]\tiny
\item[$*$] AUC denotes full ROC-AUC. ROBOSAC provides discrete operating points, so AUC is not applicable.
\end{tablenotes}
\end{threeparttable}
\caption{Late- and intermediate-fusion results. \defensename achieves a favorable detection–false-alarm trade-off across both fusion settings, with particularly strong spoofing detection.}
\label{tab:main_comparison}
\end{table}

\myparagraph{Analysis of Baseline Performance.} MATE pools evidence across objects, allowing unrelated correct reports to offset a targeted error. Spoofing adds a further difficulty: unlike an established real track, a new false track lacks a supporting history, and the attacker's own reports reinforce its estimated existence. The resulting uncertainty weakens the penalty. ROBOSAC instead tests consistency over the entire detection set. A few corrupted objects can leave that consistency above its acceptance threshold, allowing the attacker to remain in collaboration. A stricter threshold also rejects more benign collaborators. LUCIA assesses spatially pooled feature disagreement. SOMBRA's full-map perturbations produce a more distinctive discrepancy than the target-local modifications used by PGD and DF, explaining its strong SOMBRA detection but weak separation under the other attacks. \defensename addresses localized inconsistencies by excluding self-support from the reference and preserving each selected pair's evidence separately.

\myparagraph{Effect of Attack Strength.} PGD's lower attack success leaves more target reports unchanged, producing fewer contradictions for the trust estimator and explaining its lower detection rate. At attack completion, single-object removal succeeds on 63.54\%, 100\%, and 100\% of eligible independently decoded sender outputs for PGD, DF, and SOMBRA, respectively. Spoofing succeeds on 31.99\% and 80.98\% for PGD and DF. These ASRs concern whether the attack actually changes the targeted report, whereas defense metrics retain all attempted attacks, including unsuccessful ones. Appendix~\ref{subsec:attack_realization} defines the diagnostic populations.

\begin{figure}[t]
\centering
\includegraphics[width=\linewidth]{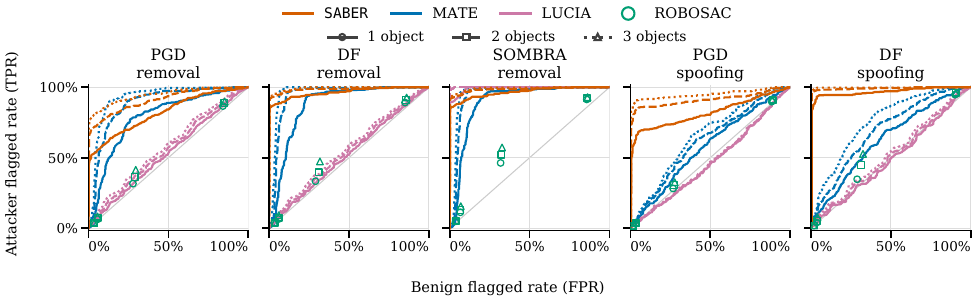}
\caption{Intermediate-fusion ROC at attack completion. \defensename achieves higher AUC than MATE and LUCIA under PGD and DF attacks, while LUCIA performs best against SOMBRA.}
\label{fig:intermediate_roc}
\end{figure}

\subsection{Component Ablations}
\label{subsec:ablation}

We evaluate four components on the full late-fusion dataset under single-object removal and spoofing attacks, using paired clean controls and the main experiment's attack interval. Each variant changes only the specified component:
\begin{enumerate}
\item \textbf{Shared agent Beta.} Replace pair-specific states and lowest-pair aggregation with one pooled agent-level Beta.
\item \textbf{Inclusive reference.} Include the assessed sender in both current and historical reference evidence.
\item \textbf{No extra amplification.} Remove the extra evidence multiplier while retaining reference-confidence weighting and negative-update bias.
\item \textbf{No memory.} Remove admission of prior contradictions while retaining the screening rule and subsequent updates.
\end{enumerate}

\myparagraph{Effect at a shared threshold.} At the same agent-trust threshold, removing pair isolation, sender exclusion, or extra amplification lowers TPR in both attack types (Table~\ref{tab:ablation}). These variants also yield fewer false alarms. This comparison shows that the components affect the flagging behavior, but does not distinguish weaker attacker--benign separation from a change in the trust-score scale.

\begin{table}[htbp]
\centering
\scriptsize
\begin{tabular}{lrrr}
\toprule
\multirow{2}{*}{Variant} & Pre-attack & \multicolumn{2}{c}{End-attack} \\
\cmidrule(lr){3-4}
 & FPR $\downarrow$ & TPR $\uparrow$ & FPR $\downarrow$ \\
\midrule
\multicolumn{4}{l}{\textit{Removal}} \\
\rowcolor{LightGreen}\defensename & 1.52 & 90.29 & 3.48 \\
Shared agent Beta & 1.52 & 69.38 & 2.03 \\
Inclusive reference & 0.63 & 5.54 & 0.56 \\
No amplification & 0.63 & 11.35 & 0.56 \\
\midrule
\multicolumn{4}{l}{\textit{Spoofing}} \\
\rowcolor{LightGreen}\defensename & 0.60 & 100.00 & 5.59 \\
Shared agent Beta & 0.30 & 74.96 & 5.28 \\
Inclusive reference & 0.00 & 0.00 & 0.00 \\
No amplification & 0.00 & 0.00 & 0.00 \\
\bottomrule
\end{tabular}
\caption{Late-fusion single-object ablations (\%) at a fixed threshold of 0.1.}
\label{tab:ablation}
\end{table}

\myparagraph{Effect at adjusted thresholds.} To test whether a different threshold can compensate for the lower detection, we allow each of the first three variants and the full method to use its own agent-trust threshold. We require at least 90\% TPR in both removal and spoofing, then select the qualifying tested threshold with the lowest worst-case FPR across the two attack types, breaking ties with the lower threshold. All other parameters remain unchanged. Even after this adjustment, each variant has higher FPR than \defensename in both attack types (Figure~\ref{fig:ablation_tradeoffs}). Thus, among the tested thresholds, the full method meets the common detection requirement with fewer benign flags. The benefit of these components is not explained solely by using the same threshold for differently scaled scores.

\begin{figure}[htbp]
\centering
\includegraphics[width=0.8\linewidth]{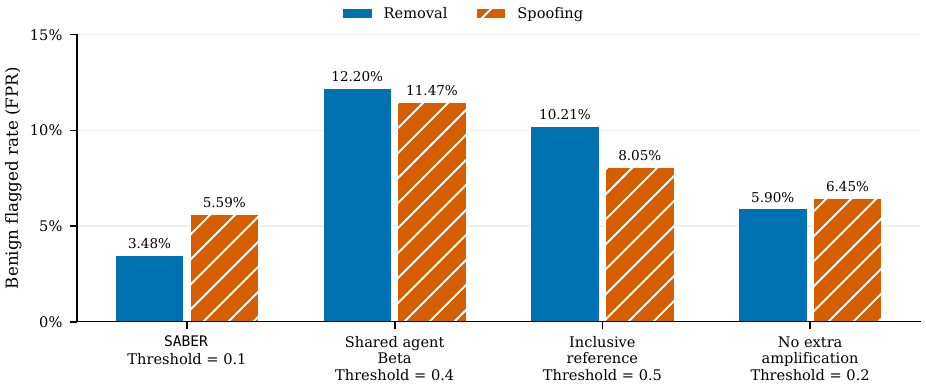}
\caption{Component ablations at thresholds achieving at least 90\% TPR in both removal and spoofing. The full method has lower FPR than each of the three variants in both attack types.}
\label{fig:ablation_tradeoffs}
\end{figure}

\myparagraph{Response with and without memory.} The fourth ablation examines response speed while keeping the screening rule and agent-trust threshold unchanged. Figure~\ref{fig:memory_ablation} shows that admitting cached contradictions makes attacker trust fall faster after screening. Memory uses evidence already observed during screening rather than requiring the pair to accumulate only new evidence after activation. It accelerates the trust response without shortening the CUSUM trigger delay.

\begin{figure}[htbp]
\centering
\includegraphics[width=0.40\linewidth]{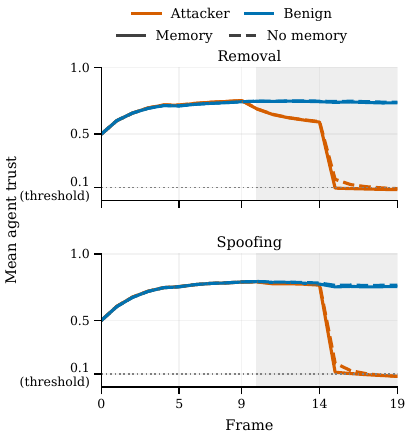}
\caption{Trust response with and without memory under single-object late-fusion attacks. Cached contradictions accelerate the attacker-trust decrease after screening.}
\label{fig:memory_ablation}
\end{figure}

\subsection{Sensitivity Studies}
\label{subsec:sensitivity}

\myparagraph{Stronger pooled penalties.} To test whether greater sensitivity in pooled assessment can address sparse attacks, we enlarge MATE's agent-level negative-update bias while keeping its threshold at 0.8 and all other settings fixed. Table~\ref{tab:ablation_sweeps} reports the results under single-object late-fusion attacks. Stronger penalties increase detection but also increase benign false alarms. Even with the largest multiplier, spoofing detection remains below \defensename while FPR is higher. At these fixed thresholds, increasing pooled-update sensitivity does not recover the full method's detection--false-alarm balance.

\begin{table}[htbp]
\centering
\scriptsize
\begin{tabular}{llrrrrrr}
\toprule
\multirow{3}{*}{Method} & \multirow{3}{*}{\makecell{Bias\\multiplier}} & \multicolumn{3}{c}{Removal} & \multicolumn{3}{c}{Spoofing} \\ \cmidrule(lr){3-5}\cmidrule(lr){6-8} &  & Pre-attack & \multicolumn{2}{c}{End-attack} & Pre-attack & \multicolumn{2}{c}{End-attack} \\  &  & FPR $\downarrow$ & TPR $\uparrow$ & FPR $\downarrow$ & FPR $\downarrow$ & TPR $\uparrow$ & FPR $\downarrow$ \\
\midrule
\rowcolor{LightGreen}\defensename & N/A & 1.52 & 90.29 & 3.48 & 0.60 & 100.00 & 5.59 \\
\multirow{3}{*}{MATE} & $1\times$ & 22.81 & 97.22 & 14.53 & 22.25 & 11.36 & 10.13 \\
 & $2\times$ & 26.91 & 99.79 & 22.58 & 24.77 & 23.05 & 18.86 \\
 & $4\times$ & 29.52 & 99.79 & 30.99 & 26.78 & 35.55 & 31.99 \\
\bottomrule
\end{tabular}
\caption{Stronger penalties of MATE on single-object late-fusion attacks (\%). MATE variants use a threshold of 0.8 and \defensename uses 0.1.}
\label{tab:ablation_sweeps}
\end{table}

\myparagraph{Attack duration.} We evaluate attacks lasting 5, 7, 10, and 20 frames. Figure~\ref{fig:attack_duration_overview} summarizes single-object attacks in both fusion settings. Late-fusion spoofing detection rises sharply between 7 and 10 frames, while intermediate-fusion spoofing detection generally improves with longer attacks. Removal detection varies less with duration, especially in intermediate fusion. Intermediate-fusion FPR stays very low throughout, whereas late-fusion FPR tends to increase with duration. The protocol and full results, including multiple targets and baselines, are provided in Appendix~\ref{sec:sensitivity_details}.

\begin{figure}[htbp]
\centering
\includegraphics[width=0.6\linewidth]{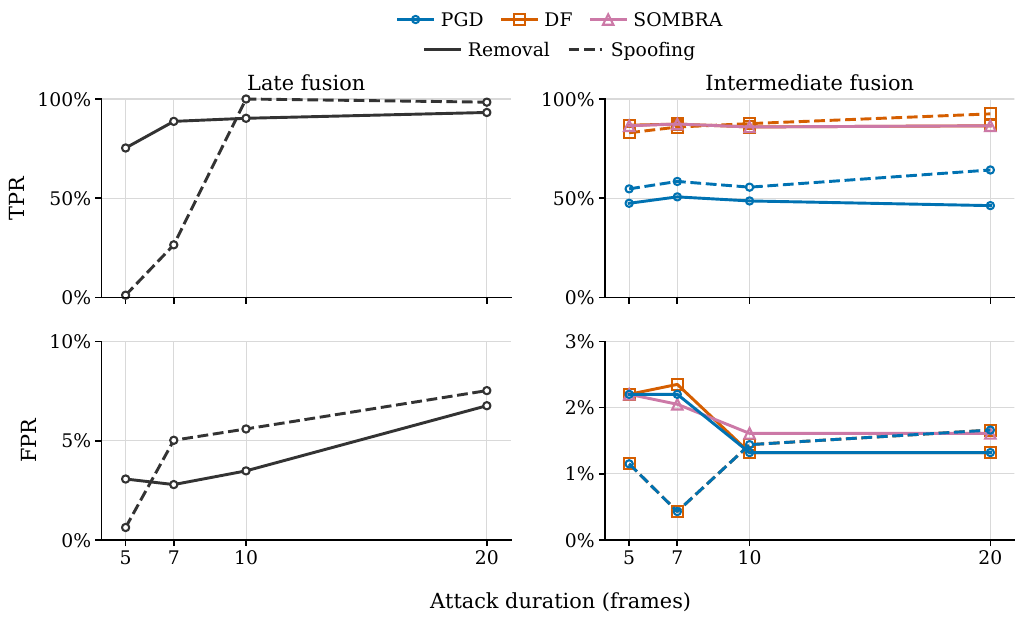}
\caption{Effect of attack duration on \defensename under single-object attacks. TPR and FPR are measured at attack completion with a fixed agent-trust threshold of 0.1. Solid and dashed lines denote removal and spoofing, respectively. Longer attacks generally improve spoofing detection, while intermediate-fusion FPR remains low.}
\label{fig:attack_duration_overview}
\end{figure}

\section{Discussion and Limitations}
\label{sec:discussion}

\myparagraph{Short or track-switching attacks.} \defensename accumulates evidence for each agent--track pair. One- or two-frame attacks may end before screening activates or sufficient evidence lowers trust. Switching targets similarly spreads contradictions across pair histories. Memory accelerates assessment within a pair but does not combine evidence across targets. As with other temporal trust estimators, we do not guarantee detection of such attacks, although individual errors may affect tier-1 trust.

\myparagraph{Insufficient benign observation coverage.} Like other consistency-based CP defenses, \defensename requires informative peer observations. If a track lies outside all benign agents' FOVs and lacks informative peer history, AGG cannot independently assess the disputed report. With only one benign observer, conflicting reports from agents with comparable trust and evidence strength can remain ambiguous. Excluding the assessed sender prevents self-support but cannot resolve this symmetry. Reliable identification therefore requires corroborating peer evidence, and FOV coverage alone does not guarantee accurate observations.

\section{Conclusion}
\label{sec:conclusion}

We presented \defensename, a two-tier Bayesian trust estimator addressing cross-object evidence dilution in collaborative perception. Isolating selected agent--track evidence and constraining agent trust through the lowest pair score prevent unrelated correct reports from directly masking localized contradictions. Our analysis gives sufficient conditions for stronger attack-induced trust reductions with bounded additional false alarms. Experiments on OPV2V show improved detection--false-alarm trade-offs for sparse removal and spoofing attacks in late and intermediate fusion.

\subsection*{Ethics Statement}

This work studies defensive trust estimation. Defense evaluations use public simulated dataset, with offline modifications to object reports or features encoded from stored sensor observations. A separate closed-loop experiment illustrates the potential safety consequences of sustained perception omissions. All vehicle control and collisions in that experiment occur within the simulator. No experiments control real vehicles, involve physical trials, or interact with deployed transportation or live network systems, so they do not expose road users or operational systems to experimental attacks. We introduce no new attack vector: removal and spoofing are established threat classes from published work, used here to assess a defense. Neither the dataset-based trust results nor the simulator illustration establish deployment readiness or guarantee collision avoidance. False alarms could suppress useful observations if such a method were deployed, while missed attacks could leave harmful information in collaboration. These risks require independent validation before any operational use.

\bibliography{references}
\bibliographystyle{plainnat}

\appendix
\section{Additional Method Definitions}
\label{sec:method_details}

\subsection{Evidence increments and weighted updates}
\label{subsec:beta}

Tier-2 scores never feed back into tier 1, whose trajectory therefore equals MATE under the same inputs \citep{hallyburton2025security}. For consistency $v\in[0,1]$, define
\begin{equation}
 n(v)=\begin{cases}\nu,&v<\theta,\\1,&v\ge\theta,\end{cases}
 \qquad \phi_+(v)=\frac{n(v)}{n(v)+1}v,\qquad
 \phi_-(v)=\frac{n(v)^2}{n(v)+1}(1-v),
 \label{eq:increments}
\end{equation}
where $\nu\ge1$ controls the negative-evidence bias and $0<\theta<1$ selects its branch. Reference and pair states can use different $\nu,\theta$. These are weighted consistency pseudo-counts, not direct observations of maliciousness.

For one eligible observation, suppress the pair and time indices and write $w=\lambda|2q-1|\ge0$ for the evidence weight in Eq.~(\ref{eq:consistency_weight}). Let $\widetilde\pi(\vartheta)$ be the propagated Beta density over pair reliability $\vartheta\in(0,1)$. The update
\begin{equation}
 \pi(\vartheta)\propto\widetilde\pi(\vartheta)
 [\vartheta^{\phi_+(v)}(1-\vartheta)^{\phi_-(v)}]^w
 \label{eq:weighted_bayes}
\end{equation}
adds $w\phi_+(v)$ and $w\phi_-(v)$ to its parameters. This is the generalized Bayesian interpretation of Eq.~(\ref{eq:tier2_update}) \citep{bissiri2016general}. Weighting the factor and weighting its increments are equivalent, so amplification is applied once.

\subsection{Screening eligibility and calibration}
\label{subsec:screen_details}

Eligibility follows Section~\ref{subsec:screening}. A matched report requires no separate visibility declaration. Invalid observation opportunities reset an untriggered CUSUM. They do not erase retained Bayesian evidence or establish recovery.

Using the same eligibility rules and candidate lifecycle as evaluation, estimate $p_0=(N_{\mathrm{event}}+1)/(N_{\mathrm{eligible}}+2)$ for each screen. Here $N_{\mathrm{eligible}}$ counts eligible clean agent--track observations and $N_{\mathrm{event}}$ counts those with $e_{ij}^t=1$. Add-one smoothing avoids zero or unit event probabilities. Let $S_{\max,n}$ denote the window maximum defined in Section~\ref{subsec:screening}, set to zero when no pair is eligible. For $N_{\mathrm{cal}}>0$ calibration windows and budget $\epsilon_{\mathrm{screen}}\in(0,1)$, select $h>0$ satisfying
\begin{equation}
 \frac{1}{N_{\mathrm{cal}}}\sum_{n=1}^{N_{\mathrm{cal}}}
 \mathbf 1\{S_{\max,n}\ge h\}\le\epsilon_{\mathrm{screen}}.
 \label{eq:screen_calibration}
\end{equation}
Move $h$ above a tied value if equality would violate this bound. Starting from $S=0$, a continuously eligible sequence with $e_{ij}^t=1$ at every frame needs
\begin{equation}
 r_*=\left\lceil\frac{h}{\log(p_1/p_0)}\right\rceil
 \label{eq:screen_delay}
\end{equation}
observations to trigger.

\subsection{Peer reference and evidence weight}
\label{subsec:references}

Write the peer reference as $q_{j,-i}=a_{j,-i}/(a_{j,-i}+b_{j,-i})$. Its parameters update as
\begin{equation}
 (a_{j,-i}^t,b_{j,-i}^t)
 =(\widetilde a_{j,-i}^t,\widetilde b_{j,-i}^t)
 +\sum_{k\in\mathcal O_{j,-i}^t}u_{kj,-i}^t
 (\phi_{+,\mathrm{ref}}(d_{kj}^t),\phi_{-,\mathrm{ref}}(d_{kj}^t)),
 \label{eq:reference_update}
\end{equation}
where $\mathcal O_{j,-i}^t$ contains eligible peers other than $i$, and $u_{kj,-i}^t=T_{1,k,-i}^{t-1}$ uses the sender-excluded broad assessment defined in Section~\ref{subsec:focused}, with its prior mean at the first frame. Propagation returns the reference toward its neutral prior.

The confidence factor $g(q)=|2q-1|\in[0,1]$ is a design choice, not an optimality claim. To analyze $w=\lambda g(q)$, let $N=\widetilde\alpha+\widetilde\beta>0$ be the pair's total Beta mass and $B=\widetilde\alpha/N$ its mean. Write $x=\phi_+(v)$ and $y=\phi_-(v)$ for the nonnegative positive- and negative-evidence increments before weighting. Holding $B,N,x,y$ fixed, the effect of $w$ is
\begin{equation}
 B_{\mathrm{new}}(w)=\frac{NB+wx}{N+w(x+y)},\qquad
 \frac{\partial B_{\mathrm{new}}}{\partial w}=\frac{N[x(1-B)-yB]}{[N+w(x+y)]^2}.
 \label{eq:weight_effect}
\end{equation}
For $x+y>0$, stronger weighting lowers trust exactly when the evidence-implied mean $x/(x+y)$ is below $B$.

\subsection{Memory admission}
\label{subsec:memory}

Let $c\in\{\mathrm{FN},\mathrm{FP}\}$ identify the omission or unmatched-report screen, with its own parameters $h_c,p_{0,c},p_{1,c}$. Set the fixed frame lookback $L_c=\lceil h_c/\log(p_{1,c}/p_{0,c})\rceil$. For agent $i$ and track $j$, let $s_{ij}^{\mathrm{start}}$ mark the start of the current track-hypothesis episode, and let $\mathcal E_{ij}(t)$ contain the distinct observation times already admitted into the pair's Beta state by time $t$. At the first screening trigger $s$, the set of past frames eligible for memory admission is
\begin{equation}
 \begin{aligned}
 \mathcal H_{ij}(s)=\big\{\ell:\;&
 \max(s_{ij}^{\mathrm{start}},s-L_c+1)\le\ell<s,\\
 &\ell\notin\mathcal E_{ij}(s^-),\quad
 I_{ij}^{\ell}=1,\ v_{ij}^{\ell}<1/2,\ w_{ij}^{\ell}>0\big\}.
 \end{aligned}
 \label{eq:memory_frames}
\end{equation}
$I_{ij}^{\ell}$ indicates evidence eligibility, and $\mathcal E_{ij}(s^-)$ is the admission set immediately before frame $s$. The memory increments in Eq.~(\ref{eq:tier2_update}) are
\begin{equation}
 M_{ij,\pm}^{s}=\sum_{\ell\in\mathcal H_{ij}(s)}
 D_{\ell,s}w_{ij}^{\ell}\phi_\pm(v_{ij}^{\ell}).
 \label{eq:memory_release}
\end{equation}
The $+$ and $-$ signs denote separate sums added to $\alpha_{ij}$ and $\beta_{ij}$, respectively. The factor $D_{\ell,t}=\prod_{r=\ell+1}^{t}(1-\rho_r)$ discounts evidence by its age, with $D_{t,t}=1$. An empty $\mathcal H_{ij}(s)$ contributes zero, and on other frames $M_{ij,\pm}^t=0$. Recovery changes assessment membership without erasing the retained state or admission history. A new hypothesis uses a new track index.

\section{Additional Experimental Settings}
\label{sec:implementation_details}

\subsection{Perception inputs}
\label{subsec:experiment_configuration}

All reports use a common world coordinate system. Reported boxes determine existence evidence and introduce unmatched candidates. We follow the visibility algorithm proposed by \citet{hallyburton2025security}. AGG accepts the resulting declarations. Paired clean and attacked runs share the same declarations, including the removal target's visible declaration during attack.

\subsection{Selection and cohort sizes}
\label{subsec:case_selection}

In addition to the selection rules in Section~\ref{subsec:evaluation_settings}, removal targets must exist throughout the window and remain within 40\,m of at least one agent. Spoofed boxes must lie within 40\,m of the attacker and be visible to it and at least two benign agents throughout the main attack interval. Table~\ref{tab:cohort_sizes} gives the cohort breakdown, including cases formed with alternative attackers or spoofing positions.

\begin{table}[htbp]
\centering
\scriptsize
\begin{tabular}{lrrrr}
\toprule
Cohort & 3 agents & 4 agents & 5 agents & Total \\
\midrule
Removal & 218 & 48 & 75 & 341 \\
Spoofing & 203 & 35 & 109 & 347 \\
\bottomrule
\end{tabular}
\caption{Full selected cohorts}
\label{tab:cohort_sizes}
\end{table}

\subsection{Baseline implementation and evaluation}
\label{subsec:baseline_implementation}

In intermediate fusion, all defenses receive the same feature messages generated against the undefended centralized model. MATE and \defensename share report associations and candidate hypotheses. ROBOSAC and LUCIA preserve their trusted-receiver assumptions. ROBOSAC samples $n-2$ teammates per trial, with budgets of 7/8/9 trials for 3/4/5-agent groups, and falls back to receiver-only perception if no consensus is reached. LUCIA retains continuous trust-weighted attention, with score thresholds used only for reporting flags.

\subsection{Metric calculation}
\label{subsec:baseline_metrics}

ROBOSAC flags tested collaborators excluded from its accepted group. An untested attacker remains in the denominator without detection credit, and receiver-only fallback (i.e. single-agent perception) rejects all peers, contributing to both TPR and FPR. ROBOSAC uses designated benign receivers at both checkpoints. LUCIA uses all receivers before attack and benign receivers during attack, excluding self-scores. Attack identities define evaluation roles only.

ROC and AUC of late-fusion and intermediate-fusion results use the same population weights, with benign peers in attacked scenes as negatives. Trust-score AUC uses complete empirical curves. ROBOSAC contributes measured operating points only.

\subsection{Attack realization diagnostics}
\label{subsec:attack_realization}

ASR measures single-object changes in the independently decoded attacker output at the attack endpoint, using BEV IoU at least 0.3. Removal succeeds when a clean-detected target disappears, and spoofing succeeds when a target absent from clean output appears. Their diagnostic denominators contain clean-detected and clean-absent targets, respectively. Defense metrics retain all selected attack attempts regardless of these outcomes.

\section{Attack-Duration Sensitivity}
\label{sec:sensitivity_details}

This appendix provides the protocol and full attack-duration results, extending the single-object summary in Figure~\ref{fig:attack_duration_overview} to multiple targets and baselines.

\subsection{Attack Duration}
\label{subsec:attack_duration_sensitivity}

Compare 5-, 7-, 10-, and 20-frame attacks at endpoints 14, 16, 19, and 29 (0-indexed), retaining all selected cases and fixed decision thresholds. The screening parameters are calibrated using clean scenarios, following Appendix~\ref{subsec:screen_details}. The memory lookback follows Appendix~\ref{subsec:memory}. Tables~\ref{tab:duration_removal}, \ref{tab:duration_spoofing}, \ref{tab:intermediate_duration_removal}, and~\ref{tab:intermediate_duration_spoofing} report all target counts and defenses. ROBOSAC remains weak against localized inconsistencies, and LUCIA remains strongest against SOMBRA but less effective against the other two attacks. These two baselines assess each endpoint independently, so their changes across durations reflect different observations and perturbations rather than accumulated history.

\begin{table}[htbp]
\centering
\scriptsize
\setlength{\tabcolsep}{2pt}
\begin{tabular}{@{}rlrrrrrrrrrr@{}}
\toprule
\multirow{3}{*}{\makecell{Duration\\(frames)}} & \multirow{3}{*}{Method} & \multirow{2}{*}{Pre-attack} & \multicolumn{9}{c}{End-attack} \\
\cmidrule(lr){4-12} & & & \multicolumn{3}{c}{1 object} & \multicolumn{3}{c}{2 objects} & \multicolumn{3}{c}{3 objects} \\
\cmidrule(lr){4-6}\cmidrule(lr){7-9}\cmidrule(lr){10-12} & & FPR $\downarrow$ & TPR $\uparrow$ & FPR $\downarrow$ & AUC $\uparrow$ & TPR $\uparrow$ & FPR $\downarrow$ & AUC $\uparrow$ & TPR $\uparrow$ & FPR $\downarrow$ & AUC $\uparrow$ \\
\midrule
\rowcolor{LightGreen}\cellcolor{white} & \defensename & \textbf{1.52} & 75.34 & 3.07 & \textbf{0.977} & 91.74 & 2.94 & \textbf{0.990} & 94.58 & 2.53 & \textbf{0.993} \\
\rowcolor{Gray}\cellcolor{white}\multirow{-2}{*}{5} & MATE & 22.81 & 99.52 & 16.54 & 0.938 & 100.00 & 17.09 & 0.963 & 100.00 & 17.67 & 0.971 \\
\midrule
\rowcolor{LightGreen}\cellcolor{white} & \defensename & \textbf{1.52} & 88.74 & 2.79 & \textbf{0.983} & 94.16 & 2.53 & \textbf{0.992} & 95.13 & 2.40 & \textbf{0.995} \\
\rowcolor{Gray}\cellcolor{white}\multirow{-2}{*}{7} & MATE & 22.81 & 99.66 & 14.31 & 0.941 & 100.00 & 15.71 & 0.967 & 100.00 & 15.93 & 0.974 \\
\midrule
\rowcolor{LightGreen}\cellcolor{white} & \defensename & \textbf{1.52} & 90.29 & 3.48 & \textbf{0.984} & 95.87 & 3.09 & \textbf{0.993} & 96.67 & 2.13 & \textbf{0.995} \\
\rowcolor{Gray}\cellcolor{white}\multirow{-2}{*}{10} & MATE & 22.81 & 97.22 & 14.53 & 0.939 & 100.00 & 15.64 & 0.966 & 100.00 & 16.12 & 0.975 \\
\midrule
\rowcolor{LightGreen}\cellcolor{white} & \defensename & \textbf{1.52} & 93.22 & 6.76 & \textbf{0.984} & 96.67 & 5.95 & \textbf{0.993} & 97.22 & 5.02 & \textbf{0.994} \\
\rowcolor{Gray}\cellcolor{white}\multirow{-2}{*}{20} & MATE & 22.81 & 88.71 & 13.82 & 0.933 & 99.20 & 14.25 & 0.970 & 100.00 & 14.55 & 0.979 \\
\bottomrule
\end{tabular}
\caption{\defensename results under late-fusion removal across attack durations. TPR and FPR are percentages. AUC denotes full ROC-AUC. MATE and \defensename use thresholds of 0.8 and 0.1.}
\label{tab:duration_removal}
\end{table}

\begin{table}[htbp]
\centering
\scriptsize
\setlength{\tabcolsep}{2pt}
\begin{tabular}{@{}rlrrrrrrrrrr@{}}
\toprule
\multirow{3}{*}{\makecell{Duration\\(frames)}} & \multirow{3}{*}{Method} & \multirow{2}{*}{Pre-attack} & \multicolumn{9}{c}{End-attack} \\
\cmidrule(lr){4-12} & & & \multicolumn{3}{c}{1 object} & \multicolumn{3}{c}{2 objects} & \multicolumn{3}{c}{3 objects} \\
\cmidrule(lr){4-6}\cmidrule(lr){7-9}\cmidrule(lr){10-12} & & FPR $\downarrow$ & TPR $\uparrow$ & FPR $\downarrow$ & AUC $\uparrow$ & TPR $\uparrow$ & FPR $\downarrow$ & AUC $\uparrow$ & TPR $\uparrow$ & FPR $\downarrow$ & AUC $\uparrow$ \\
\midrule
\rowcolor{LightGreen}\cellcolor{white} & \defensename & \textbf{0.60} & 1.24 & 0.63 & \textbf{0.948} & 2.16 & 0.63 & \textbf{0.949} & 3.08 & 0.63 & \textbf{0.950} \\
\rowcolor{Gray}\cellcolor{white}\multirow{-2}{*}{5} & MATE & 22.25 & 19.52 & 16.55 & 0.490 & 19.83 & 17.38 & 0.477 & 19.83 & 19.73 & 0.464 \\
\midrule
\rowcolor{LightGreen}\cellcolor{white} & \defensename & \textbf{0.60} & 26.59 & 5.02 & \textbf{0.949} & 27.81 & 5.02 & \textbf{0.950} & 28.32 & 5.02 & \textbf{0.950} \\
\rowcolor{Gray}\cellcolor{white}\multirow{-2}{*}{7} & MATE & 22.25 & 14.17 & 13.08 & 0.501 & 14.89 & 15.30 & 0.495 & 18.84 & 16.34 & 0.492 \\
\midrule
\rowcolor{LightGreen}\cellcolor{white} & \defensename & \textbf{0.60} & 100.00 & 5.59 & \textbf{0.958} & 100.00 & 5.59 & \textbf{0.959} & 100.00 & 5.59 & \textbf{0.959} \\
\rowcolor{Gray}\cellcolor{white}\multirow{-2}{*}{10} & MATE & 22.25 & 11.36 & 10.13 & 0.531 & 13.11 & 11.63 & 0.549 & 15.20 & 14.70 & 0.566 \\
\midrule
\rowcolor{LightGreen}\cellcolor{white} & \defensename & \textbf{0.60} & 98.41 & 7.52 & \textbf{0.976} & 100.00 & 7.52 & \textbf{0.981} & 100.00 & 7.52 & \textbf{0.982} \\
\rowcolor{Gray}\cellcolor{white}\multirow{-2}{*}{20} & MATE & 22.25 & 14.17 & 11.59 & 0.620 & 19.80 & 12.07 & 0.681 & 27.18 & 12.65 & 0.730 \\
\bottomrule
\end{tabular}
\caption{\defensename results under late-fusion spoofing across attack durations. TPR and FPR are percentages. AUC denotes full ROC-AUC. MATE and \defensename use thresholds of 0.8 and 0.1.}
\label{tab:duration_spoofing}
\end{table}

\begin{table}[p]
\centering
\scriptsize
\setlength{\tabcolsep}{2pt}
\renewcommand{\arraystretch}{0.92}
\begin{tabular}{@{}llrrrrrrrrrr@{}}
\toprule
\multirow{3}{*}{Attack} & \multirow{3}{*}{Method} & \multirow{2}{*}{Pre-attack} & \multicolumn{9}{c}{End-attack} \\
\cmidrule(lr){4-12} & & & \multicolumn{3}{c}{1 object} & \multicolumn{3}{c}{2 objects} & \multicolumn{3}{c}{3 objects} \\
\cmidrule(lr){4-6}\cmidrule(lr){7-9}\cmidrule(lr){10-12} & & FPR $\downarrow$ & TPR $\uparrow$ & FPR $\downarrow$ & AUC $\uparrow$ & TPR $\uparrow$ & FPR $\downarrow$ & AUC $\uparrow$ & TPR $\uparrow$ & FPR $\downarrow$ & AUC $\uparrow$ \\
\midrule
\multicolumn{12}{l}{\textit{5-frame attack}} \\
\rowcolor{Gray}\cellcolor{white} & ROBOSAC & 2.93 & 3.23 & 3.52 & N/A & 3.30 & 3.47 & N/A & 3.30 & 3.42 & N/A \\
 & LUCIA & 27.49 & 30.94 & 27.27 & 0.529 & 32.72 & 26.34 & 0.556 & 34.87 & 25.54 & 0.581 \\
\rowcolor{Gray}\cellcolor{white} & MATE & 27.86 & 73.31 & 27.42 & 0.762 & 87.10 & 27.76 & 0.856 & 94.43 & 28.15 & 0.910 \\
\rowcolor{LightGreen}\cellcolor{white}\multirow{-4}{*}{PGD} & \defensename & \textbf{1.27} & 47.51 & 2.20 & \textbf{0.795} & 63.93 & 2.13 & \textbf{0.884} & 75.07 & 2.13 & \textbf{0.932} \\
\cmidrule(lr){1-12}
\rowcolor{Gray}\cellcolor{white} & ROBOSAC & 2.93 & 3.37 & 3.49 & N/A & 3.23 & 3.45 & N/A & 3.23 & 3.27 & N/A \\
 & LUCIA & 27.49 & 30.35 & 27.49 & 0.521 & 31.45 & 26.37 & 0.542 & 33.38 & 25.95 & 0.561 \\
\rowcolor{Gray}\cellcolor{white} & MATE & 27.86 & 94.43 & 27.35 & 0.866 & 98.83 & 28.67 & 0.930 & 100.00 & 28.96 & 0.955 \\
\rowcolor{LightGreen}\cellcolor{white}\multirow{-4}{*}{DF} & \defensename & \textbf{1.27} & 86.80 & 2.20 & \textbf{0.965} & 92.96 & 2.13 & \textbf{0.986} & 94.13 & 2.13 & \textbf{0.994} \\
\cmidrule(lr){1-12}
\rowcolor{Gray}\cellcolor{white} & ROBOSAC & 2.93 & 4.62 & 3.52 & N/A & 4.77 & 3.45 & N/A & 5.21 & 3.45 & N/A \\
 & LUCIA & 27.49 & 92.82 & 2.81 & \textbf{0.995} & 99.12 & 0.22 & \textbf{1.000} & 100.00 & 0.00 & \textbf{1.000} \\
\rowcolor{Gray}\cellcolor{white} & MATE & 27.86 & 96.19 & 27.74 & 0.887 & 99.41 & 28.64 & 0.934 & 100.00 & 28.64 & 0.953 \\
\rowcolor{LightGreen}\cellcolor{white}\multirow{-4}{*}{SOMBRA} & \defensename & \textbf{1.27} & 86.51 & 2.20 & 0.963 & 93.26 & 2.13 & 0.986 & 94.43 & 1.98 & 0.993 \\
\midrule
\multicolumn{12}{l}{\textit{7-frame attack}} \\
\rowcolor{Gray}\cellcolor{white} & ROBOSAC & 2.93 & 3.67 & 3.64 & N/A & 3.74 & 3.59 & N/A & 3.89 & 3.59 & N/A \\
 & LUCIA & 27.49 & 31.40 & 26.76 & 0.530 & 33.16 & 26.08 & 0.556 & 34.90 & 25.54 & 0.582 \\
\rowcolor{Gray}\cellcolor{white} & MATE & 27.86 & 74.49 & 26.61 & 0.772 & 90.32 & 27.20 & 0.868 & 95.31 & 27.13 & 0.917 \\
\rowcolor{LightGreen}\cellcolor{white}\multirow{-4}{*}{PGD} & \defensename & \textbf{1.27} & 50.73 & 2.20 & \textbf{0.804} & 67.74 & 2.05 & \textbf{0.903} & 77.13 & 2.20 & \textbf{0.941} \\
\cmidrule(lr){1-12}
\rowcolor{Gray}\cellcolor{white} & ROBOSAC & 2.93 & 3.74 & 3.69 & N/A & 3.59 & 3.62 & N/A & 3.96 & 3.64 & N/A \\
 & LUCIA & 27.49 & 30.52 & 26.91 & 0.522 & 31.92 & 26.25 & 0.543 & 33.70 & 25.68 & 0.561 \\
\rowcolor{Gray}\cellcolor{white} & MATE & 27.86 & 95.01 & 27.13 & 0.872 & 99.12 & 27.49 & 0.933 & 100.00 & 27.86 & 0.957 \\
\rowcolor{LightGreen}\cellcolor{white}\multirow{-4}{*}{DF} & \defensename & \textbf{1.27} & 87.39 & 2.35 & \textbf{0.965} & 92.96 & 2.20 & \textbf{0.986} & 94.43 & 2.20 & \textbf{0.993} \\
\cmidrule(lr){1-12}
\rowcolor{Gray}\cellcolor{white} & ROBOSAC & 2.93 & 4.91 & 4.11 & N/A & 5.50 & 4.11 & N/A & 5.87 & 4.11 & N/A \\
 & LUCIA & 27.49 & 92.82 & 2.79 & \textbf{0.995} & 99.12 & 0.22 & \textbf{1.000} & 99.71 & 0.07 & \textbf{1.000} \\
\rowcolor{Gray}\cellcolor{white} & MATE & 27.86 & 96.48 & 26.91 & 0.893 & 99.12 & 27.59 & 0.936 & 100.00 & 27.76 & 0.955 \\
\rowcolor{LightGreen}\cellcolor{white}\multirow{-4}{*}{SOMBRA} & \defensename & \textbf{1.27} & 87.39 & 2.05 & 0.964 & 92.96 & 2.05 & 0.984 & 94.43 & 1.91 & 0.992 \\
\midrule
\multicolumn{12}{l}{\textit{10-frame attack}} \\
\rowcolor{Gray}\cellcolor{white} & ROBOSAC & 2.93 & 3.81 & 3.47 & N/A & 3.74 & 3.42 & N/A & 3.67 & 3.42 & N/A \\
 & LUCIA & 27.49 & 31.09 & 27.47 & 0.529 & 33.80 & 26.54 & 0.556 & 35.58 & 25.49 & 0.580 \\
\rowcolor{Gray}\cellcolor{white} & MATE & 27.86 & 76.83 & 26.03 & 0.784 & 91.50 & 26.17 & 0.881 & 96.19 & 26.69 & 0.924 \\
\rowcolor{LightGreen}\cellcolor{white}\multirow{-4}{*}{PGD} & \defensename & \textbf{1.27} & 48.68 & 1.32 & \textbf{0.819} & 67.16 & 1.17 & \textbf{0.909} & 77.42 & 1.17 & \textbf{0.945} \\
\cmidrule(lr){1-12}
\rowcolor{Gray}\cellcolor{white} & ROBOSAC & 2.93 & 3.89 & 3.52 & N/A & 3.81 & 3.49 & N/A & 3.96 & 3.49 & N/A \\
 & LUCIA & 27.49 & 30.01 & 27.66 & 0.522 & 32.55 & 26.69 & 0.544 & 34.26 & 25.95 & 0.563 \\
\rowcolor{Gray}\cellcolor{white} & MATE & 27.86 & 94.72 & 25.95 & 0.872 & 99.41 & 26.61 & 0.939 & 100.00 & 27.64 & 0.962 \\
\rowcolor{LightGreen}\cellcolor{white}\multirow{-4}{*}{DF} & \defensename & \textbf{1.27} & 85.92 & 1.32 & \textbf{0.963} & 92.38 & 1.03 & \textbf{0.986} & 93.84 & 0.88 & \textbf{0.992} \\
\cmidrule(lr){1-12}
\rowcolor{Gray}\cellcolor{white} & ROBOSAC & 2.93 & 4.84 & 3.67 & N/A & 4.99 & 3.67 & N/A & 5.72 & 3.96 & N/A \\
 & LUCIA & 27.49 & 90.91 & 3.27 & \textbf{0.993} & 98.09 & 0.59 & \textbf{1.000} & 99.56 & 0.07 & \textbf{1.000} \\
\rowcolor{Gray}\cellcolor{white} & MATE & 27.86 & 97.07 & 26.74 & 0.897 & 99.12 & 26.88 & 0.942 & 99.71 & 27.98 & 0.959 \\
\rowcolor{LightGreen}\cellcolor{white}\multirow{-4}{*}{SOMBRA} & \defensename & \textbf{1.27} & 85.92 & 1.61 & 0.965 & 92.38 & 1.17 & 0.985 & 93.84 & 0.88 & 0.991 \\
\midrule
\multicolumn{12}{l}{\textit{20-frame attack}} \\
\rowcolor{Gray}\cellcolor{white} & ROBOSAC & 2.93 & 4.25 & 3.96 & N/A & 4.33 & 3.91 & N/A & 4.62 & 4.13 & N/A \\
 & LUCIA & 27.49 & 30.33 & 27.20 & 0.526 & 33.43 & 26.42 & 0.550 & 34.97 & 24.90 & 0.573 \\
\rowcolor{Gray}\cellcolor{white} & MATE & 27.86 & 74.19 & 32.45 & 0.770 & 92.08 & 33.19 & 0.877 & 97.36 & 33.70 & 0.923 \\
\rowcolor{LightGreen}\cellcolor{white}\multirow{-4}{*}{PGD} & \defensename & \textbf{1.27} & 46.33 & 1.32 & \textbf{0.842} & 70.38 & 1.17 & \textbf{0.936} & 77.71 & 1.17 & \textbf{0.964} \\
\cmidrule(lr){1-12}
\rowcolor{Gray}\cellcolor{white} & ROBOSAC & 2.93 & 4.25 & 3.96 & N/A & 4.62 & 3.96 & N/A & 4.62 & 4.15 & N/A \\
 & LUCIA & 27.49 & 29.45 & 27.22 & 0.522 & 32.04 & 26.54 & 0.541 & 34.02 & 25.54 & 0.558 \\
\rowcolor{Gray}\cellcolor{white} & MATE & 27.86 & 89.74 & 32.50 & 0.850 & 98.53 & 33.94 & 0.930 & 100.00 & 34.68 & 0.959 \\
\rowcolor{LightGreen}\cellcolor{white}\multirow{-4}{*}{DF} & \defensename & \textbf{1.27} & 86.51 & 1.32 & \textbf{0.965} & 92.96 & 1.17 & \textbf{0.987} & 93.84 & 1.03 & \textbf{0.994} \\
\cmidrule(lr){1-12}
\rowcolor{Gray}\cellcolor{white} & ROBOSAC & 2.93 & 5.28 & 4.25 & N/A & 6.01 & 4.42 & N/A & 6.13 & 4.18 & N/A \\
 & LUCIA & 27.49 & 89.47 & 4.23 & \textbf{0.990} & 95.45 & 1.69 & \textbf{0.998} & 98.09 & 0.68 & \textbf{1.000} \\
\rowcolor{Gray}\cellcolor{white} & MATE & 27.86 & 95.01 & 33.09 & 0.880 & 98.83 & 34.58 & 0.933 & 99.71 & 35.46 & 0.956 \\
\rowcolor{LightGreen}\cellcolor{white}\multirow{-4}{*}{SOMBRA} & \defensename & \textbf{1.27} & 86.51 & 1.61 & 0.965 & 93.26 & 1.61 & 0.987 & 94.13 & 1.47 & 0.992 \\
\bottomrule
\end{tabular}
\caption{\defensename results under intermediate-fusion removal across attack durations. TPR and FPR are percentages. AUC denotes full ROC-AUC. Thresholds are fixed across attack types and target counts: \defensename 0.1, MATE 0.8, LUCIA 0.1, and ROBOSAC consistency 0.3. ROBOSAC provides discrete operating points, so full ROC-AUC is not applicable.}
\label{tab:intermediate_duration_removal}
\end{table}

\begin{table}[p]
\centering
\scriptsize
\setlength{\tabcolsep}{2pt}
\renewcommand{\arraystretch}{0.92}
\begin{tabular}{@{}llrrrrrrrrrr@{}}
\toprule
\multirow{3}{*}{Attack} & \multirow{3}{*}{Method} & \multirow{2}{*}{Pre-attack} & \multicolumn{9}{c}{End-attack} \\
\cmidrule(lr){4-12} & & & \multicolumn{3}{c}{1 object} & \multicolumn{3}{c}{2 objects} & \multicolumn{3}{c}{3 objects} \\
\cmidrule(lr){4-6}\cmidrule(lr){7-9}\cmidrule(lr){10-12} & & FPR $\downarrow$ & TPR $\uparrow$ & FPR $\downarrow$ & AUC $\uparrow$ & TPR $\uparrow$ & FPR $\downarrow$ & AUC $\uparrow$ & TPR $\uparrow$ & FPR $\downarrow$ & AUC $\uparrow$ \\
\midrule
\multicolumn{12}{l}{\textit{5-frame attack}} \\
\rowcolor{Gray}\cellcolor{white} & ROBOSAC & \textbf{0.00} & 1.08 & 1.10 & N/A & 1.15 & 1.08 & N/A & 1.15 & 1.08 & N/A \\
 & LUCIA & 24.87 & 24.50 & 28.00 & 0.465 & 24.64 & 27.86 & 0.470 & 25.36 & 27.86 & 0.474 \\
\rowcolor{Gray}\cellcolor{white} & MATE & 22.25 & 25.65 & 22.60 & 0.523 & 27.95 & 23.39 & 0.533 & 28.53 & 25.34 & 0.540 \\
\rowcolor{LightGreen}\cellcolor{white}\multirow{-4}{*}{PGD} & \defensename & \textbf{0.00} & 54.76 & 1.15 & \textbf{0.826} & 77.23 & 1.01 & \textbf{0.925} & 84.15 & 1.01 & \textbf{0.958} \\
\cmidrule(lr){1-12}
\rowcolor{Gray}\cellcolor{white} & ROBOSAC & \textbf{0.00} & 1.15 & 1.10 & N/A & 1.30 & 1.10 & N/A & 1.51 & 1.08 & N/A \\
 & LUCIA & 24.87 & 25.79 & 27.31 & 0.487 & 27.43 & 26.80 & 0.509 & 29.01 & 25.96 & 0.530 \\
\rowcolor{Gray}\cellcolor{white} & MATE & 22.25 & 30.26 & 22.81 & 0.570 & 39.19 & 24.35 & 0.613 & 47.26 & 26.13 & 0.655 \\
\rowcolor{LightGreen}\cellcolor{white}\multirow{-4}{*}{DF} & \defensename & \textbf{0.00} & 83.00 & 1.15 & \textbf{0.973} & 95.10 & 1.01 & \textbf{0.991} & 96.83 & 0.86 & \textbf{0.996} \\
\midrule
\multicolumn{12}{l}{\textit{7-frame attack}} \\
\rowcolor{Gray}\cellcolor{white} & ROBOSAC & \textbf{0.00} & 1.59 & 1.61 & N/A & 1.59 & 1.61 & N/A & 1.59 & 1.61 & N/A \\
 & LUCIA & 24.87 & 24.21 & 28.07 & 0.466 & 24.64 & 27.79 & 0.471 & 24.64 & 27.47 & 0.475 \\
\rowcolor{Gray}\cellcolor{white} & MATE & 22.25 & 24.21 & 21.30 & 0.545 & 28.82 & 22.43 & 0.560 & 33.43 & 24.62 & 0.584 \\
\rowcolor{LightGreen}\cellcolor{white}\multirow{-4}{*}{PGD} & \defensename & \textbf{0.00} & 58.50 & 0.43 & \textbf{0.853} & 75.50 & 0.43 & \textbf{0.933} & 85.59 & 0.43 & \textbf{0.969} \\
\cmidrule(lr){1-12}
\rowcolor{Gray}\cellcolor{white} & ROBOSAC & \textbf{0.00} & 1.59 & 1.63 & N/A & 1.59 & 1.66 & N/A & 1.59 & 1.66 & N/A \\
 & LUCIA & 24.87 & 25.50 & 26.99 & 0.489 & 27.52 & 26.20 & 0.511 & 30.02 & 25.50 & 0.532 \\
\rowcolor{Gray}\cellcolor{white} & MATE & 22.25 & 31.70 & 21.47 & 0.587 & 41.50 & 24.50 & 0.643 & 49.86 & 26.44 & 0.692 \\
\rowcolor{LightGreen}\cellcolor{white}\multirow{-4}{*}{DF} & \defensename & \textbf{0.00} & 85.88 & 0.43 & \textbf{0.979} & 96.54 & 0.43 & \textbf{0.994} & 97.98 & 0.43 & \textbf{0.998} \\
\midrule
\multicolumn{12}{l}{\textit{10-frame attack}} \\
\rowcolor{Gray}\cellcolor{white} & ROBOSAC & \textbf{0.00} & 1.59 & 1.63 & N/A & 1.59 & 1.61 & N/A & 1.59 & 1.61 & N/A \\
 & LUCIA & 24.87 & 24.78 & 28.24 & 0.468 & 24.93 & 27.95 & 0.473 & 25.36 & 27.95 & 0.477 \\
\rowcolor{Gray}\cellcolor{white} & MATE & 22.25 & 27.09 & 21.13 & 0.575 & 35.45 & 23.01 & 0.616 & 42.94 & 23.90 & 0.652 \\
\rowcolor{LightGreen}\cellcolor{white}\multirow{-4}{*}{PGD} & \defensename & \textbf{0.00} & 55.62 & 1.44 & \textbf{0.819} & 76.66 & 1.44 & \textbf{0.910} & 85.01 & 1.44 & \textbf{0.947} \\
\cmidrule(lr){1-12}
\rowcolor{Gray}\cellcolor{white} & ROBOSAC & \textbf{0.00} & 1.87 & 1.87 & N/A & 1.87 & 1.87 & N/A & 1.87 & 1.87 & N/A \\
 & LUCIA & 24.87 & 27.09 & 27.38 & 0.491 & 29.32 & 26.22 & 0.513 & 31.20 & 25.29 & 0.533 \\
\rowcolor{Gray}\cellcolor{white} & MATE & 22.25 & 34.29 & 20.70 & 0.629 & 46.40 & 22.72 & 0.693 & 61.10 & 25.34 & 0.760 \\
\rowcolor{LightGreen}\cellcolor{white}\multirow{-4}{*}{DF} & \defensename & \textbf{0.00} & 87.61 & 1.44 & \textbf{0.966} & 96.25 & 1.30 & \textbf{0.989} & 98.27 & 1.30 & \textbf{0.993} \\
\midrule
\multicolumn{12}{l}{\textit{20-frame attack}} \\
\rowcolor{Gray}\cellcolor{white} & ROBOSAC & \textbf{0.00} & 2.45 & 2.45 & N/A & 2.52 & 2.45 & N/A & 2.67 & 2.47 & N/A \\
 & LUCIA & 24.87 & 23.63 & 28.31 & 0.471 & 24.35 & 28.17 & 0.475 & 24.35 & 28.17 & 0.479 \\
\rowcolor{Gray}\cellcolor{white} & MATE & 22.25 & 45.82 & 27.52 & 0.636 & 56.48 & 28.82 & 0.701 & 65.99 & 30.86 & 0.748 \\
\rowcolor{LightGreen}\cellcolor{white}\multirow{-4}{*}{PGD} & \defensename & \textbf{0.00} & 64.27 & 1.66 & \textbf{0.888} & 81.84 & 1.66 & \textbf{0.958} & 91.07 & 1.66 & \textbf{0.985} \\
\cmidrule(lr){1-12}
\rowcolor{Gray}\cellcolor{white} & ROBOSAC & \textbf{0.00} & 2.45 & 2.45 & N/A & 2.67 & 2.50 & N/A & 3.03 & 2.45 & N/A \\
 & LUCIA & 24.87 & 26.22 & 26.90 & 0.496 & 28.60 & 25.91 & 0.517 & 30.79 & 25.24 & 0.538 \\
\rowcolor{Gray}\cellcolor{white} & MATE & 22.25 & 46.11 & 27.74 & 0.669 & 65.71 & 28.89 & 0.769 & 78.67 & 30.91 & 0.835 \\
\rowcolor{LightGreen}\cellcolor{white}\multirow{-4}{*}{DF} & \defensename & \textbf{0.00} & 92.51 & 1.66 & \textbf{0.990} & 98.27 & 1.51 & \textbf{0.997} & 99.14 & 1.37 & \textbf{0.997} \\
\bottomrule
\end{tabular}
\caption{\defensename results under intermediate-fusion spoofing across attack durations. TPR and FPR are percentages. AUC denotes full ROC-AUC. Thresholds are fixed across attack types and target counts: \defensename 0.1, MATE 0.8, LUCIA 0.1, and ROBOSAC consistency 0.3. ROBOSAC provides discrete operating points, so full ROC-AUC is not applicable.}
\label{tab:intermediate_duration_spoofing}
\end{table}

\section{Additional Results and Operating Trade-offs}
\label{sec:additional_results}

\paragraph{Separation across operating points.} \defensename has higher ROC-AUC than MATE under every evaluated attack family and target count, with particularly large gains on late-fusion spoofing (Table~\ref{tab:main_comparison}). Its low-FPR detection also improves under PGD and DF, while LUCIA achieves the strongest separation against SOMBRA (Figures~\ref{fig:late_roc_low} and~\ref{fig:intermediate_roc_low}).

\begin{figure}[htbp]
\centering
\includegraphics[width=0.85\linewidth]{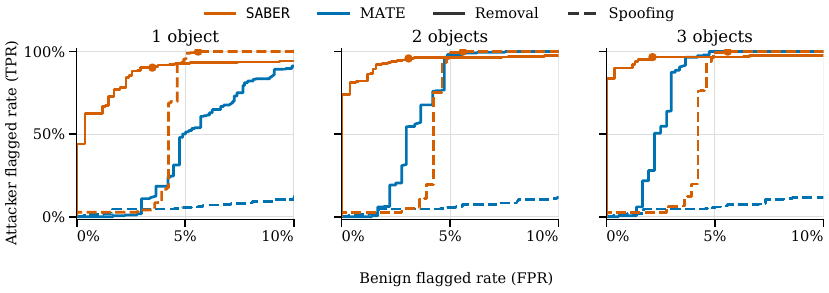}
\caption{Late-fusion ROC in the 0--10\% FPR region.}
\label{fig:late_roc_low}
\end{figure}

\begin{figure}[htbp]
\centering
\includegraphics[width=\linewidth]{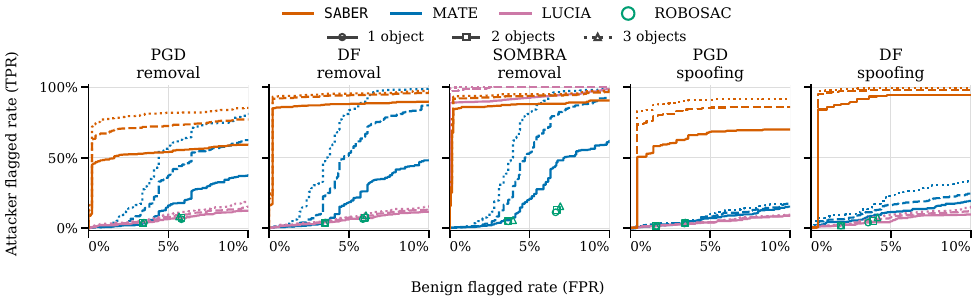}
\caption{Intermediate-fusion ROC in the 0--10\% FPR region.}
\label{fig:intermediate_roc_low}
\end{figure}

\FloatBarrier
\section{Theoretical Analysis and Proofs}
\label{sec:theory_statements}

Following Section~\ref{subsec:theoretical_properties}, we first establish protection against cross-object dilution, then analyze trust response through weighted evidence, memory, and paired clean--attack comparisons, and finally derive attacker--benign separation and false-alarm bounds. The attack endpoint is fixed at $t_\star=t_e$ before observing outcomes. All evidence sums include admitted memory and temporal discounts. Unless time is shown explicitly, scores and evidence are evaluated at $t_\star$. The results apply to the report histories in either fusion setting under their stated assumptions.

\subsection{Accumulated evidence and dilution}
\label{subsec:accumulation}

We establish the pair-isolation property and contrast it with pooled assessment. Using the admission set defined in Appendix~\ref{subsec:memory}, define the pre-amplification evidence
\begin{equation}
 U_{ij}=\sum_{\ell\in\mathcal E_{ij}(t_\star)}D_{\ell,t_\star}g(q_{j,-i}^{\ell})\phi_+(v_{ij}^{\ell}),\qquad
 V_{ij}=\sum_{\ell\in\mathcal E_{ij}(t_\star)}D_{\ell,t_\star}g(q_{j,-i}^{\ell})\phi_-(v_{ij}^{\ell}).
 \label{eq:evidence_sums}
\end{equation}
Here $U_{ij}$ and $V_{ij}$ accumulate the pair's admitted positive and negative evidence, including memory and temporal discounts but excluding the multiplier $\lambda$. Subtracting the fixed initial parameters from Eq.~(\ref{eq:propagation}) leaves a discounted deviation plus admitted increments. Historical evidence released at $s$ has total discount $D_{\ell,s}D_{s,t_\star}=D_{\ell,t_\star}$. Induction gives $\alpha_{ij}=\kappa\mu_i+\lambda U_{ij}$ and $\beta_{ij}=\kappa(1-\mu_i)+\lambda V_{ij}$ without double-counting. The resulting pair trust at the checkpoint is
\begin{equation}
 B_{ij}^{t_\star}=\frac{\kappa\mu_i+\lambda U_{ij}}{\kappa+\lambda(U_{ij}+V_{ij})},
 \label{eq:unrolled}
\end{equation}

\begin{proof}[Proof of Proposition~\ref{prop:isolation_summary}]
Under the fixed histories in the proposition, Eq.~(\ref{eq:evidence_sums}) contains only the assessed pair's evidence. Direct updates for other tracks cannot change its mean. Taking the minimum in Eq.~(\ref{eq:pair_aggregate}) gives $T_{2,i}\le B_{ij}$ while $j\in\mathcal A_i^{t_\star}$. Equation~(\ref{eq:combine}) also gives $R_i\le B_{ij}$.
\end{proof}

\begin{proposition}[Cross-object dilution under pooled assessment]
\label{prop:dilution}
A pooled Beta state with initial parameters $(A_0,B_0)>0$, disputed evidence $(x,y)\ge0$, and $m$ unrelated evidence blocks $(b_+,b_-)\ge0$ has mean
\begin{equation}
 T_{\mathrm{pool}}(m)=\frac{A_0+x+mb_+}{A_0+B_0+x+y+m(b_++b_-)}.
 \label{eq:background}
\end{equation}
If $b_++b_->0$ and $b_+/(b_++b_-)>T_{\mathrm{pool}}(0)$, this mean increases with $m$ and converges to $b_+/(b_++b_-)$.
\end{proposition}

\begin{proof}
Differentiation in the nonnegative real extension of $m$ gives numerator $b_+(B_0+y)-b_-(A_0+x)$ over a positive squared denominator. The numerator is positive under the stated condition, so the result also holds for integer block counts. Dividing numerator and denominator by $m$ gives the limit.
\end{proof}

The comparison isolates direct cancellation: unrelated evidence raises pooled trust but cannot replenish an assessed pair's state. Other reports may still affect the peer reference or initial mean, which are held fixed in Proposition~\ref{prop:isolation_summary}. Minimum aggregation does not add penalties across attacked objects, so several equally contradicted pairs can yield the same agent score.

\subsection{Weighted evidence and detection by the attack endpoint}
\label{subsec:finite_detection}

For $0<\eta<1$, rearranging Eq.~(\ref{eq:unrolled}) gives the exact pair-level crossing condition
\begin{equation}
 B_{ij}<\eta
 \quad\Longleftrightarrow\quad
 \lambda\big[\eta V_{ij}-(1-\eta)U_{ij}\big]
 >\kappa(\mu_i-\eta).
 \label{eq:pair_crossing}
\end{equation}
For an assessed pair this is sufficient to flag its agent, since $R_i\le B_{ij}$. The left side is the pair's net evidence toward crossing the threshold. The right side is the resistance from its initial state. Neither contains direct evidence from unrelated tracks.

For $U_{ij}+V_{ij}>0$, hold initialization and the reference, screening, admission, discount, and assessment histories fixed. Increasing $\lambda$ then lowers $B_{ij}$ exactly when $U_{ij}/(U_{ij}+V_{ij})<\mu_i$, by differentiation as in Eq.~(\ref{eq:weight_effect}), and
\begin{equation}
 \lim_{\lambda\to\infty}B_{ij}
 =\frac{U_{ij}}{U_{ij}+V_{ij}}.
 \label{eq:weight_limit}
\end{equation}
Thus amplification increases the influence of supported contradictions, but cannot create a more negative evidence ratio. If both the initial mean and this ratio are at least $\eta$, no finite weight makes the pair cross below it. With no evidence, $B_{ij}=\mu_i$ regardless of $\lambda$.

To connect Eq.~(\ref{eq:pair_crossing}) to persistent errors, consider admitted observations with $0\le v\le v_a<\min(\theta,1/2)$. Binary reports imply $g(q)=|2v-1|$, so valid increment bounds are
\begin{equation}
 g(q)\phi_+(v)\le x_a=\frac{\nu v_a}{\nu+1},\qquad
 g(q)\phi_-(v)\ge y_a=\frac{(1-2v_a)\nu^2(1-v_a)}{\nu+1}.
 \label{eq:attack_increments}
\end{equation}
These use the tier-2 bias parameters. They follow from $g(q)\le1$ and $g(q)=1-2v\ge1-2v_a$ on this interval.

\begin{proposition}[Finite-time detection from one contradicted pair]
\label{prop:finite_detection}
Suppose an attacker pair $(a,j)$ remains assessed at $t_\star$. Let $\mathcal C\subseteq\mathcal E_{aj}(t_\star)$ be admitted observations satisfying the consistency condition above, with effective length $L=\sum_{\ell\in\mathcal C}D_{\ell,t_\star}>0$. Let $U_0,V_0$ be the discounted pre-amplification evidence from all other admitted observations of this pair. Then $R_a<\eta$ if
\begin{equation}
 \lambda\big\{\eta(V_0+Ly_a)
 -(1-\eta)(U_0+Lx_a)\big\}
 >\kappa(\mu_a-\eta).
 \label{eq:finite_detection}
\end{equation}
\end{proposition}

\begin{proof}
Equation~(\ref{eq:attack_increments}) yields $U_{aj}\le U_0+Lx_a$ and $V_{aj}\ge V_0+Ly_a$. Substituting these bounds into the left side of Eq.~(\ref{eq:pair_crossing}) proves $B_{aj}<\eta$. Assessment membership then gives $R_a\le B_{aj}<\eta$.
\end{proof}

When $\eta y_a-(1-\eta)x_a>0$, Eq.~(\ref{eq:finite_detection}) can be written directly as a sufficient effective-evidence requirement:
\begin{equation}
 L>
 \frac{\kappa(\mu_a-\eta)/\lambda-\eta V_0+(1-\eta)U_0}
 {\eta y_a-(1-\eta)x_a}.
 \label{eq:effective_evidence_requirement}
\end{equation}
The denominator is a lower bound on net threshold-crossing evidence per effective observation. The numerator accounts for the prior and the pair's other admitted evidence. Thus the condition specifies how much supported contradiction must be admitted by the attack endpoint, rather than requiring an unspecified duration of persistent errors.

Screening determines when evidence can first be admitted. For a continuously eligible attack producing an event at every frame from $t_s$, starting with zero CUSUM, the trigger is $s=t_s+r_*-1$, with $r_*$ from Eq.~(\ref{eq:screen_delay}). If $s>t_\star$, this attack interval alone cannot activate tier 2. Otherwise, suppose the pair remains assessed from $s$ through $t_\star$, every observation satisfies the consistency bound, and $\rho$ is constant. Their effective length is $\sum_{k=0}^{t_\star-s}(1-\rho)^k$. Memory adds the discounts of qualifying admitted pre-trigger observations, using their original references. Other observations remain in $U_0,V_0$. The relevant quantity in Eq.~(\ref{eq:effective_evidence_requirement}) is discounted evidence, not the raw number of attack frames. With constant $\rho>0$, even indefinitely many distinct observations have total effective length at most $1/\rho$, so persistence alone does not guarantee crossing.

\subsection{Effect of admitted memory}
\label{subsec:memory_effect}

The following result isolates how admitting cached evidence changes trust, without changing initialization or subsequent updates.

\begin{proposition}[Admitting history under fixed initialization and subsequent updates]
\label{prop:memory}
Hold a pair's initialization, propagation, and non-memory increments fixed. At a common checkpoint, let its no-memory state have mean $B_{\mathrm{base}}$ and mass $N>0$. If the additional historical increments remaining after discounting are $(M_+,M_-)$, with $M=M_++M_->0$, then
\begin{equation}
 B_{\mathrm{mem}}-B_{\mathrm{base}}=\frac{M}{N+M}\left(\frac{M_+}{M}-B_{\mathrm{base}}\right).
 \label{eq:memory_effect}
\end{equation}
\end{proposition}

\begin{proof}
The memory state has mean $(NB_{\mathrm{base}}+M_+)/(N+M)$. Subtract $B_{\mathrm{base}}$. Linear propagation preserves the parameter difference as the discounted historical increment.
\end{proof}

Memory lowers pair trust when the admitted history has a lower evidence-implied mean than the no-memory state. It lowers agent trust only when the changed pair affects the minimum in Eq.~(\ref{eq:combine}). This response can occur at activation without shortening the CUSUM delay. The proposition isolates history admission under fixed initialization. Runs with different initial means must retain their own initialization in the full comparison. Figure~\ref{fig:memory_ablation} evaluates the resulting response speed.

\subsection{Paired reduction relative to MATE}
\label{subsec:paired}

This comparison asks whether tier 2 strengthens the attack-induced trust decrease, rather than merely lowering scores on both clean and attacked inputs. We compare against MATE because tier 1 of \defensename refers to its design. This comparison isolates the additional attack-specific response introduced by selective pair-level assessment, rather than differences between unrelated trust estimators. Couple clean and attacked runs with the same sensor observations, visibility declarations, and initialization rules. Each run keeps its own resulting screening and trust history. Within either run, tier 1 equals MATE on those inputs. Let $P_i=T_{1,i}-R_i\ge0$ denote the additional reduction due to tier 2. For the selected agent, suppress its index and define
\begin{equation}
 D_{\mathrm M}=T_1^{\mathrm{clean}}-T_1^{\mathrm{attack}},\qquad
 D=R^{\mathrm{clean}}-R^{\mathrm{attack}},\qquad
 D-D_{\mathrm M}=P^{\mathrm{attack}}-P^{\mathrm{clean}}.
 \label{eq:paired_identity}
\end{equation}
Here $D$ and $D_{\mathrm M}$ are the attack-induced trust decreases for \defensename and MATE, respectively. The last equality follows by substitution and removes any extra reduction that tier 2 already causes on the paired clean input.

\paragraph{Sufficient evidence bounds.} We first bound the extra reduction on the attacked input from below and its cost on the paired clean input from above. For an assessed attacker pair, suppose $\mu_a\le\overline\mu_a\in(0,1)$, $U_{aj}\le\overline U_a$, $V_{aj}\ge\underline V_a$, and $T_{1,a}\ge m_a$, with nonnegative evidence bounds and $0<c<m_a<1$. The condition
\begin{equation}
 \lambda[(m_a-c)\underline V_a-(1-m_a+c)\overline U_a]
 \ge\kappa(\overline\mu_a-m_a+c)
 \label{eq:attack_condition}
\end{equation}
ensures $B_{aj}\le m_a-c$ by Eq.~(\ref{eq:unrolled}), hence $P_a\ge c$. If these conditions hold with probability at least $1-\xi$, they give $\Pr(P^{\mathrm{attack}}\ge c)\ge1-\xi$. Failures of assessment membership, informative evidence, or the initial/baseline bounds are included in $\xi$.

For a clean run with nonempty assessment, write $\mu=\mu_i$ and define $C_{\max}=\max_{j\in\mathcal A_i^{t_\star}}[\mu(U_{ij}+V_{ij})-U_{ij}]_+$, where $[z]_+=\max(0,z)$. Equations~(\ref{eq:unrolled}) and~(\ref{eq:combine}) give
\begin{equation}
 P^{\mathrm{clean}}\le\min\{1,[T_1-\mu]_++\lambda C_{\max}/\kappa\}.
 \label{eq:clean_cost}
\end{equation}
To verify this, use $[\mu-B_{ij}]_+\le\lambda[\mu(U_{ij}+V_{ij})-U_{ij}]_+/\kappa$ and $\mu-T_2=\max_{j\in\mathcal A_i^{t_\star}}[\mu-B_{ij}]_+$. Conditional on nonempty clean assessment, assume $[T_1-\mu]_+\le d_0$ and $C_{\max}\le c_0$ jointly with probability at least $1-\zeta$, where $d_0,c_0\ge0$ and $\zeta\in[0,1]$. Then $C_{\mathrm{clean}}=(1-\zeta)\min\{1,d_0+\lambda c_0/\kappa\}+\zeta$ is a valid bound. These sums include memory, so the clean-side effect of weighting and history admission is not omitted.

\paragraph{From evidence bounds to paired reduction.} The attack-side evidence condition gives $c$, while the clean-side bound gives $C_{\mathrm{clean}}$. The following theorem combines these with their occurrence probabilities.

\begin{theorem}[Larger paired trust reduction]
\label{thm:paired}
Suppose the clean assessment set is nonempty at $t_\star$ with probability at most $\gamma$. Conditional on this event, when it has positive probability, assume the expected additional reduction is at most $C_{\mathrm{clean}}\in[0,1]$. Suppose $P^{\mathrm{attack}}\ge c>0$ with probability at least $1-\xi$. For $\gamma,\xi\in[0,1]$,
\begin{equation}
 \mathbb E[D-D_{\mathrm M}]
 \ge(1-\xi)c-\gamma C_{\mathrm{clean}}.
 \label{eq:paired_bound}
\end{equation}
Thus $(1-\xi)c>\gamma C_{\mathrm{clean}}$ suffices for a larger expected attack-induced reduction than MATE.
\end{theorem}

\begin{proof}
Nonnegativity gives $\mathbb E[P^{\mathrm{attack}}]\ge(1-\xi)c$. Without assessed pairs, $P^{\mathrm{clean}}=0$, so $\mathbb E[P^{\mathrm{clean}}]\le\gamma C_{\mathrm{clean}}$. Apply Eq.~(\ref{eq:paired_identity}). No independence or equal activation histories are required.
\end{proof}

\subsection{Attacker--benign separation}
\label{subsec:separation}

We now compare the attacker with a benign peer in the same attacked scene. Write $f(\mu,U,V)=(\kappa\mu+\lambda U)/[\kappa+\lambda(U+V)]$. For $\kappa>0$, $\lambda\ge1$, $0<\mu<1$, and nonnegative evidence, it increases with $\mu,U$ and decreases with $V$.

\begin{theorem}[Separation from one contradicted pair]
\label{thm:separation}
At $t_\star$, suppose attacker $a$ has an assessed pair with $\mu_a\le\overline\mu_a$, $U_{aj}\le\overline U_a$, and $V_{aj}\ge\underline V_a$. For benign peer $b$, suppose $T_{1,b}\ge m_b$ and, whenever its assessment set is nonempty, $\mu_b\ge\underline\mu_b$, $U_{bj}\ge\underline U_b$, and $V_{bj}\le\overline V_b$ for every assessed pair. Assume $m_b,\overline\mu_a,\underline\mu_b\in(0,1)$ and nonnegative evidence bounds. Set
\begin{equation}
 u_a=f(\overline\mu_a,\overline U_a,\underline V_a),\qquad
 \ell_b=\min\{m_b,\underline\mu_b,f(\underline\mu_b,\underline U_b,\overline V_b)\}.
 \label{eq:separation_bounds}
\end{equation}
Then $R_a\le u_a$, $R_b\ge\ell_b$, and $R_b-R_a\ge\ell_b-u_a$. If the attacker and joint benign conditions fail with probabilities at most $\xi_a,\xi_b\in[0,1]$, the gap bound holds with probability at least $1-\xi_a-\xi_b$.
\end{theorem}

\begin{proof}
Monotonicity bounds the attacker's assessed pair, and hence $R_a$, above by $u_a$. It bounds every assessed benign pair below by the third argument of $\ell_b$. The other arguments bound tier 1 and the initial-mean cap. If no benign pair is assessed, $R_b=T_{1,b}\ge m_b$. Subtraction and a union bound prove the claims.
\end{proof}

The benign assumptions allow errors but bound their retained contradiction evidence jointly over assessed pairs. Only one attacker pair needs a sufficiently strong contradiction. These are assumptions on reporting histories, not consequences of labeling an agent benign or malicious. When $\ell_b>u_a$, the theorem gives strict separation. If pre-attack agent trusts are equal to $r_0$, then $(r_0-R_a)-(r_0-R_b)=R_b-R_a$, so this gap equals the difference in their trust decreases. Equal starting scores are needed only for this interpretation.

For attacker and benign peer within one attacked scene, the analogous identity is
\begin{equation}
 (R_b-R_a)-(T_{1,b}-T_{1,a})=P_a-P_b.
 \label{eq:gap_improvement}
\end{equation}
Thus $P_a\ge c_a>c_b\ge P_b$ improves their separation relative to MATE. Pure-clean runs and benign peers in attacked scenes are different populations, so assumptions on their reductions are stated separately.

\subsection{Detection with bounded additional false alarms}
\label{subsec:false_alarms}

A low attacker score is useful only if a threshold can detect it without excessive benign flags. In this comparison, $\mathrm{TPR}$, $\mathrm{FPR}$, and $\eta$ refer to \defensename, while the subscript $\mathrm M$ denotes MATE. Fix thresholds $0<\eta\le\eta_{\mathrm M}<1$ and a specified benign population. Using the tier-2 reduction $P_i$ defined in Appendix~\ref{subsec:paired} and suppressing the benign agent index, the FPR difference is
\begin{equation}
 \begin{aligned}
 \mathrm{FPR}-\mathrm{FPR}_{\mathrm M}
 &=\Pr_b(T_1\ge\eta_{\mathrm M},P>T_1-\eta)\\
 &\quad-\Pr_b(T_1<\eta_{\mathrm M},P\le T_1-\eta)\\
 &\le\Pr_b(T_1\ge\eta_{\mathrm M},P>\eta_{\mathrm M}-\eta).
 \end{aligned}
 \label{eq:relative_fpr}
\end{equation}
The first term counts newly introduced false alarms and the second counts removed false alarms. Dropping the latter and using $T_1\ge\eta_{\mathrm M}$ gives the inequality.

\begin{proposition}[A joint detection--false-alarm operating point]
\label{prop:threshold_interval}
Suppose $R_a\le u_a$ with probability at least $1-\xi_a$, and $\Pr_b(T_1\ge\eta_{\mathrm M},P>p_b)\le\delta$, where $u_a,\xi_a,\delta\in[0,1]$ and $p_b\ge0$. If $u_a<\eta_{\mathrm M}-p_b$, every fixed threshold
\begin{equation}
 u_a<\eta\le\eta_{\mathrm M}-p_b
 \label{eq:threshold_interval}
\end{equation}
satisfies $\mathrm{TPR}\ge1-\xi_a$ and $\mathrm{FPR}\le\mathrm{FPR}_{\mathrm M}+\delta$.
\end{proposition}

\begin{proof}
The lower threshold bound flags every outcome with $R_a\le u_a$. The upper bound gives $\eta_{\mathrm M}-\eta\ge p_b$, so Eq.~(\ref{eq:relative_fpr}) bounds additional FPR by $\delta$.
\end{proof}

The same threshold must satisfy both requirements, excluding the vacuous choice of a zero threshold. This establishes a useful detection--false-alarm operating point under the stated assumptions, not superiority to every baseline at every threshold. As a separate comparison, if $\Pr_a(T_{1,a}\ge m_a)\ge1-\zeta_a$ and $\eta\le m_a$, then $\mathrm{TPR}_{\mathrm M}(\eta)\le\zeta_a$, since $\{T_{1,a}<\eta\}\subseteq\{T_{1,a}<m_a\}$. This last observation concerns MATE at the same low threshold, not at its own $\eta_{\mathrm M}$.

\paragraph{A benign evidence condition.} The minimum over pairs requires control of the most adverse pair, rather than average error frequency. Let $\mathcal H=\{T_{1,b}\ge\eta_{\mathrm M},\mathcal A_b^{t_\star}\ne\varnothing\}$ have probability at most $\gamma_H\in[0,1]$. For this bound, relabel agent $b$'s pair states by their order of creation, with at most $K_b<\infty$ states per window. Choose $\eta<\mu_{\min}<1$ and define
\begin{equation}
 \mathcal F_{bj}=\left\{j\in\mathcal A_b^{t_\star},\quad
 \eta V_{bj}-(1-\eta)U_{bj}
 >\frac{\kappa}{\lambda}(\mu_{\min}-\eta)\right\}.
 \label{eq:benign_tail_condition}
\end{equation}
For $\varepsilon_0,\ldots,\varepsilon_{K_b}\in[0,1]$, assume, when $\Pr_b(\mathcal H)>0$, that
\begin{equation}
 \Pr_b(\mu_b<\mu_{\min}\mid\mathcal H)\le\varepsilon_0,\qquad
 \Pr_b(\mathcal F_{bj}\mid\mathcal H)\le\varepsilon_j,\quad j=1,\ldots,K_b.
 \label{eq:benign_tail_assumptions}
\end{equation}
For an uncreated pair, $\mathcal F_{bj}$ is empty. These assumptions allow occasional benign contradictions while bounding the chance that they accumulate enough retained evidence to cross the decision threshold. They imply
\begin{equation}
 \mathrm{FPR}-\mathrm{FPR}_{\mathrm M}
 \le\gamma_H\min\{1,\varepsilon_0+\textstyle\sum_{j=1}^{K_b}\varepsilon_j\}.
 \label{eq:benign_tail_bound}
\end{equation}
Indeed, a newly flagged benign agent must belong to $\mathcal H$. Unless its initial mean is below $\mu_{\min}$, crossing requires an assessed pair with $B_{bj}<\eta$. Equation~(\ref{eq:pair_crossing}) then implies $\mathcal F_{bj}$. A conditional union bound proves Eq.~(\ref{eq:benign_tail_bound}) without pair independence. If $\Pr_b(\mathcal H)=0$, no additional false alarm is introduced.

\end{document}